%% file: main.tex
\documentclass[aps,physrev,reprint,twocolumn,longbibliography,amsmath,amssymb,superscriptaddress]{revtex4-2}
\pdfoutput=1
\usepackage[utf8]{inputenc}
\usepackage[english]{babel}
\usepackage[T1]{fontenc}
\usepackage{xcolor}
\usepackage{amsmath}
\usepackage{amsthm}
\usepackage{amssymb}
\usepackage{amsfonts}
\usepackage{graphicx}
\usepackage{bbold}
\usepackage{tikz}
\usetikzlibrary{
 arrows.meta,
 positioning,
 shapes.geometric,
 decorations.pathreplacing,
 calc,
 fit,
 backgrounds
}

\newtheorem{lemma}{Lemma}
\newtheorem*{lemma*}{Lemma}
\newtheorem{proposition}{Proposition}

\newcommand{\ket}[1]{| #1 \rangle}
\newcommand{\bra}[1]{\langle #1 |}
\newcommand{\braket}[2]{\langle #1 \mid #2 \rangle}
\newcommand{\HH}{\mathcal{H}}

\newtheorem{corollary}[lemma]{Corollary}
\theoremstyle{remark}

\newcommand{\rhoini}{\rho_{\mathrm{ini}}}
\newcommand{\rhotar}{\rho_{\mathrm{tar}}}
\newcommand{\rhofin}{\rho_{\mathrm{fin}}}

\usepackage{hyperref}

\newcommand{\orcid}[1]{}

\begin{document}

\title{Quantum Circuit Design via Dynamic Pauli Constraints}

\author{James R. Wootton}
\orcid{0000-0003-1943-5306}
\affiliation{Moth, Arlesheim, Switzerland}

\author{Merlin Incerti-Medici}
\orcid{0000-0001-8404-9036}
\affiliation{Moth, Arlesheim, Switzerland}

\author{Daniel Bultrini}
\orcid{0000-0002-2534-4200}
\affiliation{Moth, Arlesheim, Switzerland}

\author{Pierre Fromholz}
\orcid{0000-0001-6822-2337}
\affiliation{Moth, Arlesheim, Switzerland}

\begin{abstract}
We introduce the Motte model, a software-oriented model of quantum computation motivated by the practical constraints of near-term quantum
hardware. In this model, gates are specified by constraints expressed in terms of Pauli observables, with each disjoint layer of gates accompanied by a pairwise or $k$-local quantum state tomography of the device.
We prove that the model is equivalent to the coupling-graph-restricted
circuit model, and hence universal for BQP, with only polynomial overhead: with three-body tomography on a coupling graph of bounded degree, emulating a depth-$D$ circuit on $N$ qubits
costs $\mathcal{O}(D^2 N \log N)$ elementary operations. Because the gate applied at each step is fixed exactly by the tomographic report, this emulation is robust to both the sampling and readout noise of the tomography. The model formalizes an idiom shared by existing work that ranges from quantum imaginary time evolution for the study of quantum systems to the use of quantum computers for procedural
generation in games. It therefore provides a natural interface for designing quantum software entirely in terms of physically observable quantities, relevant for the NISQ era and into fault-tolerance, with gate and decoherence noise lying outside the present scope.
\end{abstract}

\maketitle

\section{Introduction}
\label{sec:intro}

Quantum hardware has now been available from cloud services for a decade~\cite{cross2017openqasm}, allowing anyone to submit a quantum circuit to a wide range of QPUs. Nevertheless, these will remain limited to near-term \emph{NISQ} hardware for the next few years~\cite{preskill2018nisq,ibmroadmap2025,quantinuumroadmap2024}. Noisy gates and limited gate connectivity therefore remain important constraints for the design of quantum software.

One response is to base algorithms on what NISQ devices do \emph{naturally}: quantum dynamics. Qubits are quantum systems that evolve and interact; if a task can be framed as a dynamical process on qubits, it can already be run on real hardware with the reasonable expectation that noise will not destroy the results~\cite{kim2023utility}. This is especially true when we work with the natural observables of the device: Pauli expectation values for single qubits or directly coupled groups of qubits on the QPU.

\begin{figure}[t!]
 \centering
 \input{fig-motte-schematic.tex}
 \caption{Generic structure of the
 Motte model. The \emph{user} works only in
 Pauli expectations; the classical \emph{Motte interface} deduces the matching
 gate; the \emph{QPU} executes it on the coupling graph~$G$. Toward an
 \emph{end target}, each layer measures the $\le k$-body Paulis (to precision
 $\pm\delta$), the user names new $k$-body targets (optionally a fractional
 power $f$ of the deduced unitary), and the interface appends the gate. The
 arrows circulate: the loop runs for $D$ layers, from the \emph{initial state}
 $\ket{0}^{\otimes N}$ to a \emph{final} $Z$-basis \emph{sampling}.}
 \label{fig:motte-schematic}
\end{figure}
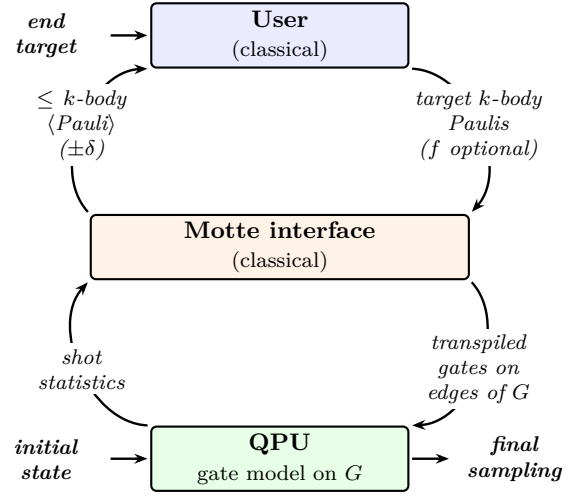

This principle underlies multiple independent lines of work, with two serving as the primary inspiration. One is the non-variational algorithm for quantum imaginary time evolution (QITE)~\cite{motta2020qite}, in which ground and excited states of a Hamiltonian are found by repeatedly measuring Pauli expectation values and applying correction unitaries derived from them to drive toward a target state. The other is \texttt{QuantumGraph}~\cite{wootton2020map,fromholz2026backrooms}, a framework that builds quantum algorithms around pairwise tomography of a multi-qubit register and the manipulation of Pauli expectation values. In this approach, the user never specifies or receives state vectors. Instead, they work entirely with quantities that can actually be estimated from hardware: $\langle X \rangle$, $\langle Y \rangle$, $\langle Z \rangle$ per qubit, and their pairwise correlators. Gates are defined by where these values should go rather than by what matrix should be applied.

Additionally, there is the family of feedback-based quantum algorithms, beginning with FALQON~\cite{magann2022falqon} and extended to ground-state preparation~\cite{larsen2023feedback}. In these, each successive layer's gate parameter is set deterministically from a Pauli-observable measurement on the current state, with no classical optimization loop. Finally, there is the shadow-enhanced greedy quantum eigensolver (SEGQE)~\cite{erle2026segqe}, in which classical shadows of the current state are used in classical post-processing to estimate the energy reduction induced by each candidate gate from a discrete pool, with the gate of largest estimated reduction greedily appended at each step.

All four of these frameworks share the same underlying idiom: observe Pauli expectation values, compare with the target, deduce and apply to the state the gates that get them there. Here we introduce a model of quantum computation that formalizes this idiom, the \emph{Motte model}, illustrated in Fig.~\ref{fig:motte-schematic}. We establish its computational standing, and argue for its particular suitability as an interface for creative quantum software.

The motivation for this model is the idea that it could provide a more intuitive framework for the design of certain quantum algorithms. It is therefore not a model in the sense of a novel approach toward designing quantum hardware, but rather for the design of quantum software. Applications of this model vary as the motivational examples show, from studies of quantum systems via QITE~\cite{motta2020qite}, to creative applications such as procedural content generation for computer games~\cite{wootton2020map,fromholz2026backrooms,Seskir2022}.

This motivation is founded on the idea that local Pauli expectation values are a more manageable set of variables than a multi-qubit statevector~\cite{huang2020shadows}, and driving toward target values is a more intuitive heuristic for designing software than the application of rotations. Game designers do not think in $SU(4)$ when managing the stats of NPCs in a game~\cite{Seskir2022}. Musicians might ask producers to ``make this louder and more distorted'' rather than specifying a quantum channel. In such cases, a more imperative form of algorithm design can be useful, leaving the system to translate these commands into quantum operations. The Motte model formalizes exactly this interface.

This paper is organized as follows. The model is defined in Sec.~\ref{sec:model}, presenting it as a novel interface for quantum software. We detail the core iterative process behind the interface in which each step begins with a local Pauli tomography of the circuit so far, with the results provided to the user. The user then specifies new target values for selected Pauli expectation values, which implicitly define the next layer of operations applied to the circuit. In Sec.~\ref{sec:equiv} we then address the computational power of the model. We show that it is sufficiently expressive to emulate the circuit model with only polynomial overhead. This demonstrates that the Motte model is fully capable of any quantum computation, but is otherwise not the intended, nor a likely useful way of exploiting the model. In Sec.~\ref{sec:applications} we expand on how QITE and \texttt{QuantumGraph} are used in practice. Unlike the constructions used to establish universality, which demonstrate computational power through an atypical use of the model, these examples illustrate the intended use of the interface.

\section{The Motte Model}
\label{sec:model}

\subsection{Definition of the model} \label{sec:def}

The Motte model is defined on a QPU of $N$ qubits with a connected coupling graph $G=(V,E)$, $|V|=N$. At the beginning of each computation, the $N$ qubits are initialized to the state $|0\rangle^{\otimes N}$. The QPU is assumed to be noiseless. The device, by which we mean the QPU together with the classical control software required to provide the Motte model interface, is accessed as follows.

\begin{enumerate}

\item \textbf{Tomography:} A $k$-body tomography of all qubits can be performed at any time, with the resulting $\leq k$-body Pauli expectation values reported to the user (or more often, a classical process defined by the user). This is done up to a user-specified statistical noise (sampling error) $\delta$. The number of distinct non-identity Pauli expectation values of weight at most $k$ on $N$ qubits is $\sum_{j=1}^{k}\binom{N}{j}3^j$, but the number of measurement \emph{settings} required to estimate them is far smaller, since Paulis supported on disjoint regions commute and share a setting. The setting count and the shots per setting are derived in Sec.~\ref{sec:overhead}.

\item \textbf{Iterative circuit construction:} Gates are specified by providing a set of at most $k$-body target Pauli expectation values. Intuitively, the user names where the measured Pauli values should go, and the device supplies the gate that takes the current state there. Formally, the targets constrain the $\le k$-body reduced density matrices (RDMs) that the state would have after the iteration, and the iteration $U$ is a unitary consistent with that constraint: it carries the RDMs reconstructed from the stated tomography to RDMs meeting the targets. The constraint does not need to single out $U$; how a given realization resolves the residual freedom, and when that freedom disappears altogether, is the subject of Sec.~\ref{sec:realization}. As it provides more control, the iteration uses the principal matrix power $U^f$ instead when a fractional power $f\in[0,1]$ is provided. It naturally does not reach the target when the fraction is only applied once. Using the fractional power and update between iterations allows better control of the desired convergence towards the target at the cost of extra resources (more steps) and the guarantee of never reaching the target exactly with a finite number of steps. Two-qubit gates are restricted to adjacent pairs $(j,k)\in E$.

\item \textbf{Final Measurement:} Projective measurement of all qubits in the basis of choice, typically the $Z$ basis. Mid-circuit measurement is not supported.

\end{enumerate}

\subsection{Realization of the model} \label{sec:realization}

There is no natural physical system that implements the above model. Nevertheless, it is worth noting that two approaches do come close to a physical implementation. One is imaginary-time evolution driven by quantum nondemolition measurements, a theoretical proposal which closes a measure-decide-apply loop using weak collective measurements~\cite{kondappan2022}. The other is quantum feedback control~\cite{wiseman2010}, experimentally realized in cavity quantum electrodynamics~\cite{sayrin2011}, which more generally conditions a Hamiltonian drive on a continuous measurement record.

We therefore consider that there is always a standard gate model quantum computer with the same number of qubits and coupling map under the hood, with the Motte model simply providing the interface. Within this context, we can then detail the specific implementations of the Motte model paradigm, as well as their complexity.
We distinguish below two abstract realizations of the model that are crucially distinct, and then discuss how QITE fits between the two.

Common to both abstractions, the hardware-native gates (most often single- and two-qubit gates only) applied at each step are determined jointly by the user-supplied target expectation values and the results of the most recent tomography. From the $k$-body expectation values one can construct all $\leq k$-qubit RDMs and vice-versa. 

\subsubsection{Gates via eigenstate transformation} \label{sec:gates_eigenstates}

For this abstraction, the transformation applied in one iteration is deduced from the tomographed and target RDMs by pairing their eigenvectors in a unique and deterministic manner.
  Assuming both RDMs non-degenerate for now, a $k$-qubit unitary that maps the $j$th eigenstate of the former to the $j$th eigenstate of the latter is unique and well-defined up to the choice of phases for each vector of the initial eigenbasis. This $k$-qubit unitary is deduced classically and then transpiled to get the circuit for the iteration in terms of hardware-native gates. As long as the two RDMs have the same eigenvalues, each single step reaches the target RDM. The phase ambiguity, however, changes the correlations between the $k$ qubits and the rest of the register. After several iterations involving different subsets of qubits, the compound effect of the phases not being fixed generally leads to failure in reaching the target final state despite targeting a priori realizable RDMs along the way. Independently from that fact, the transformation for a single step is always unitary, such that a step cannot change the spectrum (in particular the purity) of the RDM, even if the target would require it.

To illustrate the phenomenon with a simple case let's take the initial state $\rhoini = \alpha \ket{0}\bra{0} + \beta \ket{1}\bra{1}$ and $\rhotar = \alpha \ket{+}\bra{+} + \beta \ket{-}\bra{-}$, then any $U(\theta) = \ket{+}\bra{0} + e^{i \theta}\ket{-}\bra{1}$ applied on the initial state reaches the target despite the phase $\theta$ not being fixed. Take now $\ket{\psi} = \sqrt{\alpha}\ket{00}+ \sqrt{\beta}\ket{11}$. The initial RDM of the first qubit is the same as above, and taking the same target will get us $U(\theta)$, for one $\theta$ chosen arbitrarily by the algorithm. The final state is $\ket{\varphi} = \sqrt{\alpha}\ket{+0}+ \sqrt{\beta}e^{i \theta}\ket{-1}$ such that e.g. $\langle Z\otimes X\rangle = 2\sqrt{\alpha \beta}\cos \theta$. Observables now depend on the uncontrolled angle.

This ambiguity becomes harmless when both the initial and target states are each pure and a product across the $k$-qubit subset and the rest and for a full step ($f=1$). Despite the added degeneracy coming from the purity of the states, each of the many possible unitaries obtained with the above procedure will correctly transform the initial state to the target state, up to a global phase. Any systematic choice from these possible unitaries is fine as long as it is unique and reproducible. For fractional steps, it is not equivalent however: different choices of $U$ have different principal powers, so $U^f$ takes the state along different paths. Some paths will be more advantageous than other (depth-wise or convergence-wise). In \texttt{QuantumGraph}, for example, each eigenstate from a degenerate subspace of the tomography-inferred RDM is then projected onto the matching degenerate subspace of the target RDM and the projected set is reorthonormalized using the Gram-Schmidt procedure (and seeded random vector when needed). The order of the lists of initial and target states is then unique and defines the relevant unitary.

When constructing gates this way, the user-supplied expectation values can be regarded as a constraint on the existing state rather than a uniquely specified target state. For example, for a target of $\langle Z\otimes Z \rangle=+1$, the two dominant eigenstates (those with the highest eigenvalues) would be rotated to the closest possible states that exist within this subspace. The new state would then be as aligned to the $\langle Z\otimes Z \rangle=+1$ target as possible.

\subsubsection{Gates via density matrix transformation} \label{sec:gates_density}
As discussed above, the main issue with gates via eigenstate transformation is the phase ambiguity. However, this can be resolved by working one level higher in the tomographic hierarchy, with $m$-qubit gates determined with a tomography on $k=m+1$ qubits. This form of gate is primarily motivated by the arguments to be made in the next section, and so a full specification of this form of gate can be found in Sec.~\ref{subsec:dens}. There the freedom is not resolved by convention but removed outright: for generic reported RDMs the constraint admits a single $U$, recovered from a linear system without either RDM being diagonalized.

\subsubsection{Example: QITE}

A pioneer realization of the gate-specification procedure arises from QITE~\cite{motta2020qite}.
Consider a geometric $\ell$-local Hamiltonian $H = \sum_m h_m$, where each of the $M$ terms $h_m$ acts on at most $\ell$ neighboring qubits of the coupling graph $G$.
Ground and thermal states of $H$ are approached through imaginary-time evolution
$e^{-\beta H}\ket{\Psi}/\|e^{-\beta H}\ket{\Psi}\|$,
which is approximated via a first-order Trotter decomposition into $n = \beta/\Delta\tau$ rounds, each consisting of one application of each non-unitary factor $e^{-\Delta\tau h_m}$ (a step).
The QITE algorithm replaces each such factor with a unitary derived from Pauli tomography of the current state~\cite{motta2020qite}, instantiating precisely the measure-decide-apply idiom of the Motte model: the user supplies a local Hamiltonian term $h_m$ and the number of steps $n$ acting as the inverse of the fraction $f$ of Sec.~\ref{sec:def}, and the interface determines and applies the gate that advances the state one imaginary-time step.

The gate applied at each step is a unitary $e^{-i\Delta\tau A_m}$ supported on a \emph{domain} of $k \geq \ell$ qubits surrounding the support of $h_m$, where
\begin{equation*}
 A_m = \sum_{I} a_m^{(I)}\, \sigma_I
\end{equation*}
is a Hermitian operator expanded in the Pauli basis $\{\sigma_I\}$ restricted to those $k$ qubits. The existence of $A_m$ is guaranteed by Uhlmann's theorem~\cite{uhlmann1976} if $k = \mathcal{O}(C^d)$, with $C$ a bound on the correlation length of every state along the imaginary-time trajectory on a $d$-dimensional lattice~\cite{hastings2006}.

The coefficients $a_m^{(I)}$ are determined classically by solving a dependent linear system involving only the Pauli expectation values on the $k$-qubit domain around the support of $h_m$: the $k$-(body) tomography. The kernel of this linear system is often relatively high dimensional in practice, such that one of many unitary solutions is chosen at (seeded) random. This fact alone brings QITE closer to the abstraction of Sec.~\ref{sec:gates_eigenstates}.

Then, the circuit for $e^{-i\Delta\tau A_m}$ only involves the $k$ qubits in the domain.

Unlike the two abstractions, the target of QITE is only implicitly set by the user through $h_m$ or $H$, but, like the other two, a single step reconstruction only involves a succession of $k$-qubit unitaries completely deduced from a preceding $k$-tomography.

When the available $k$ is smaller than $\ell$ or the correlation length due to hardware constraints, the linear system is solved over a truncated Pauli basis and the resulting gate is no longer an exact imaginary-time step.
This corresponds to the \emph{inexact} QITE regime of Ref.~\cite{motta2020qite}: $k = 1$ yields a mean-field-like approximation, and increasing $k$ gives successively better approximations.
This truncation is compatible with the Motte model: a computation with fixed $k < k_{\mathrm{exact}}$ remains internally consistent and still drives the state toward lower energy overall, though the energy is no longer guaranteed to decrease at every step~\cite{motta2020qite} and the trajectory is no longer the exact imaginary-time one.

\subsection{Differences with variational algorithms}

As an iterative and adaptive approach to evolve toward a target, the Motte model has surface similarities to variational quantum algorithms (VQAs) such as the variational quantum eigensolver \cite{peruzzo2014variational} and the Quantum Approximate Optimization Algorithm (QAOA) \cite{farhi2014quantum}. All rely on a hybrid quantum-classical procedure to minimize a global cost function.

In general, VQAs utilize a parameterized circuit ansatz of depth $D$ which is run multiple times to evaluate the cost function and potentially also parameter gradients. A classical optimizer is then used to iteratively update gate parameters to minimize the cost function. In contrast, though the Motte model could similarly be considered as `optimizing' unitaries, they are constructed analytically rather than variationally through the mapping of the tomography to the Pauli constraints.

Furthermore, in VQAs the optimization over the often high-dimensional and global parameter space is hindered by issues such as complex optimization landscapes and barren plateaus. These issues are also often worse in the presence of noise \cite{larocca2024review} and under the limitation of finite sampling, which then causes further issues for the classical optimizer. The Motte model instead trades the advantages and disadvantages of a classical optimizer and parameter-dependent complexity in favor of a complexity that increases with the number of iterations. This relocates rather than removes the difficulty: in place of barren plateaus, the analytic gate construction is sensitive to the conditioning of the tomographic RDMs, which in the $k=2$ construction degrades as their eigenvalue gaps $\gamma_{\min}$ shrink (App.~\ref{app:k2overhead}); the exact $k=3$ route is free of this particular dependence, though its linear system has a conditioning sensitivity of its own (App.~\ref{app:reconstructing.U}).

The question of whether VQAs could show quantum advantage is also still open~\cite{cerezoDoesProvableAbsence2025}, even though the ansatz behind QAOA is universal when the evolution times of the Hamiltonians are used as explicitly programmed instructions rather than the heuristic optimization of parameters~\cite{lloyd2018quantum}. The Motte model, by contrast, already subsumes QITE which is an instance aimed at ground- and thermal-state preparation for classically intractable Hamiltonians, a regime where quantum simulation is broadly expected to be advantageous~\cite{feynman1982simulating,lloyd1996universal} and where, for states of bounded correlation length, QITE achieves an exponential per-iteration reduction in space and time over exact classical imaginary-time evolution~\cite{motta2020qite}. As we now show, the model also admits this programmed-instruction universality.

\section{Equivalence to the Circuit Model}
\label{sec:equiv}
The Motte model is useful as a model of quantum computation only if it has computational power equivalent to the standard circuit model. To demonstrate that this is the case, we will now show that the Motte model on any connected coupling graph $G$ is computationally equivalent to the standard quantum circuit model on $G$~\cite{wootton2024circuit}, with polynomial overhead. This is done by showing that the circuit model can efficiently emulate the Motte model and vice-versa. Since SWAP routing allows any connected graph to emulate any other with polynomial overhead (topology-dependent, $\mathcal{O}(N)$ in the worst case and polylogarithmic on well-connected graphs~\cite{Beals2013}), the Motte model is universal for BQP on any connected graph.

\subsection{Emulation of the Motte model by the circuit model} \label{sec:emulation_m2c}

The Motte model is imagined as an interface through which quantum computation can be performed, and a novel framework for constructing quantum software, rather than a new form of quantum hardware. As such, we consider that any Motte model device is a circuit model quantum computer under the hood. The circuit model can therefore emulate the Motte model by construction.

The assumption that a model of quantum computing is in fact the circuit model under the hood is not unique to the Motte model. It follows naturally from the fact that noiseless gates are assumed for BQP, as well as the fact that all known practically realizable fault-tolerant methods use the circuit model with quantum error correction. Any route to a true quantum computer therefore requires the circuit model, whatever paradigm for quantum software it uses. This is most evident for the topological model, where the non-abelian anyons are excitations of a topological error-correcting code, created and manipulated by explicit circuits on that code~\cite{Chen2025}. For measurement-based quantum computation, fault-tolerance can be realized via emulation of the circuit model with a surface code based approach~\cite{Raussendorf2007}. For the adiabatic model, the requirement to maintain a gap that can become arbitrarily small is not compatible with noisy systems~\cite{Sarandy2005}, but emulation via a fault-tolerant circuit would be straightforward~\cite{farhi2000quantumcomputationadiabaticevolution}. In all cases, fault-tolerance comes from emulation by or emulation of the circuit model.

\subsection{Emulation of the circuit model by the Motte model} \label{subsec:emulation.of.circuit.by.Motte}

We now determine if and when the Motte model interface is compatible with universal quantum computation, and the complexity overhead required. We adopt a strategy of `proof by hack'. This is a proof by construction in which we suppose that we were given a Motte model device but need to run a standard quantum circuit. How can we most straightforwardly hack the Motte model to emulate the circuit model?

Any circuit model computation can be represented in terms of a universal gate set of single- and two-qubit gates on a connected coupling graph~\cite{barenco1995elementary}. We can therefore restrict to these. Given a circuit of such gates, the question is then how to implement them through the Motte model interface.

Within the majority of this section, the process considered concerns three parties: the user, the underlying circuit model quantum computer, and the classical computer that interfaces between the quantum computer and the user. The latter two together constitute the Motte model device. Adding gates to the circuit consists of the following steps.
\begin{enumerate}
 \item The quantum computer runs a tomography of the circuit so far and supplies the results to the user.
 \item The user chooses the method with which the interface computer will deduce the required gates, i.e. gates via eigenstate transformation, via density-matrix transformation, or via QITE (Sec.~\ref{sec:realization}), typically the same for each step.
 \item Using their own classical computer, the user reconstructs the matrix representation of the $k$-qubit RDM for all qubits that will be targeted by their desired gate $U$.
 \item The user's classical computer then applies the matrix representation of $U$ to the RDM and deduces the resulting expectation values. These expectation values are the new target supplied to the interface computer. This requires only classical computation on matrices of fixed size, at a cost of $\mathcal{O}(1)$ per gate, as accounted for in Sec.~\ref{sec:overhead}.
 \item The interface computer reconstructs the matrix representation of both the initial and target RDMs and, from that pair alone, deduces a unitary $\tilde{U}$ carrying one to the other. How it does so depends on the method chosen in step 2: for gates via eigenstate transformation it diagonalizes both RDMs and uses the eigenvalues to identify pairs of corresponding initial and target eigenstates; for gates via density matrix transformation it instead solves the single linear system of App.~\ref{app:reconstructing.U}, and no diagonalization is involved.
 \item The interface computer transpiles the matrix representation of $\tilde{U}$ to quantum gates and adds it to the circuit.
\end{enumerate}
The details of how the Motte model emulates the circuit model depend on the method chosen by the user for the interface computer to reconstruct $U$, as well as the value of $k$.

We will first consider the case of gates via eigenstate transformation alone, and show that there is no straightforward path to emulation. We will then show that gates via density matrix transformation alone are sufficient for a tomography with $k=3$. Combining gates via different methods at $k=2$ is also sufficient to give universality, but requires an extension of the model. Since this extension is not strictly within the bounds of the Motte model interface as defined earlier, we leave this to App.~\ref{app:combined}. Finally, we will consider a small extension of standard QITE to best fit within the Motte model paradigm, provide an explicit example of its implementation and discuss its universality.

\subsubsection{Gates via eigenstate transformation} \label{sec:emulation_eigenstates}

Here we consider gates via eigenstate transformation at $k=2$, which is the current implementation of the \texttt{QuantumGraph} package. We will show that this does not by itself suffice for universality due to the phase ambiguity already illustrated in Sec.~\ref{sec:gates_eigenstates}. This will then motivate the constructions of the sections that follow.

As described above, single- and two-qubit gates can be added by taking the given single- or two-qubit tomography, deducing the resulting RDM and then transforming it via the matrix representation of the desired gate.

By diagonalizing both the initial and target RDMs, the interface computer will determine the initial eigenstates $\ket{\alpha_j}$ and target eigenstates $\ket{\beta_j}$. However, note that $\ket{\beta_j} = e^{i \theta_j} U \ket{\alpha_j}$: the deduced target eigenstates do not simply reflect the action of $U$, but will also have an arbitrary phase $e^{i \theta_j}$ due to the phase ambiguity when recovering from a density matrix. Equivalently we can say that
\begin{equation*}
 \ket{\beta_j} = U V \ket{\alpha_j} \,\,\,\, \forall j,
\end{equation*}
where $V = \sum_j e^{i \theta_j} \ket{\alpha_j}\bra{\alpha_j}$ is a unitary diagonal in the initial eigenbasis.

Using the method of eigenstate transformation, the interface computer therefore constructs the proposed unitary $\tilde U$ with
\begin{equation*}
 \tilde U = \sum_j e^{i \theta_j} \ket{\beta_j}\bra{\alpha_j} = UV.
\end{equation*}
This is not $U$, but $U$ up to the effect of arbitrary phases on each eigenstate. Since this will not be the desired operation in general, these ambiguous phases lead to an error which prevents a straightforward emulation of the circuit model. Whether more complex methods to achieve an emulation are possible in this case, or for gates via eigenstate transformation with higher values of $k$, is left for future work.

\subsubsection{Gates via density matrix transformation}
\label{subsec:dens}

The method of gates via density matrix transformation was specifically designed to resolve the phase ambiguity. As such, it is ideal for the case of a user who wishes to hack the Motte model into applying a predefined unitary. As seen above, the action of an $m$-qubit unitary on $2^m$ orthogonal rays (rather than vectors) is not enough to reconstruct $U$. However, we can instead consider the action of $U \otimes \mathbb{1}$ on $m+1$ qubits, where $\mathbb{1}$ is the single-qubit identity. As we will now show, the action of this unitary on an $(m{+}1)$-qubit RDM allows us to fully reconstruct $U$. Throughout, $\rhoini$ will be used to denote the $(m{+}1)$-qubit RDM reported by the tomography and
\begin{equation} \label{eq:rhotar.main}
 \rhotar = (U \otimes \mathbb{1})\, \rhoini \, (U \otimes \mathbb{1})^{\dagger}
\end{equation}
its image under the intended gate, which is the target the user supplies. We write $n = 2^m$, so that both are $2n \times 2n$ matrices, with $2n = 2^{k}$ at tomography size $k = m+1$. We then have the following two results.

\begin{lemma} \label{lem:uniquely.determining.U}
 If there exists no proper subspace $\{0\} \neq E \subsetneq \mathbb{C}^{n}$ such that $\rhoini$ sends $E \otimes \mathbb{C}^2$ and $E^{\perp} \otimes \mathbb{C}^2$ to themselves, equivalently, if $\rhoini$ is not block-diagonal with respect to any decomposition $(E \otimes \mathbb{C}^2) \oplus (E^{\perp} \otimes \mathbb{C}^2)$, a property we will call \emph{block-irreducibility}, then $\rhoini$ and $\rhotar$ determine $U$ up to a global phase.
\end{lemma}

\begin{lemma} \label{lem:determining.U.is.generic}
 The set of hermitian matrices that fail to be block-irreducible, that is, that preserve $E \otimes \mathbb{C}^2$ and $E^{\perp} \otimes \mathbb{C}^2$ for some proper subspace $\{0\} \neq E \subsetneq \mathbb{C}^n$, has Lebesgue measure zero. Block-irreducibility, and hence the hypothesis of Lemma~\ref{lem:uniquely.determining.U}, therefore holds generically.
\end{lemma}

We prove these two results in App.~\ref{app:reconstructing.U}. Here, we simply observe that this means that for a block-irreducible RDM $\rhoini$, which by Lemma~\ref{lem:determining.U.is.generic} is the generic case, we can reconstruct $U$ from the RDM and its image after conjugating with $U \otimes \mathbb{1}$. If we now follow the steps outlined at the start of Sec.~\ref{subsec:emulation.of.circuit.by.Motte}, the precise reconstruction of $U$ allows us to transpile the matrix representation of $U$ to quantum gates and add it to the circuit.

As stated above, this unique reconstruction requires $\rhoini$ to be block-irreducible, and that is not a restrictive demand in the context of the equivalence. Failure is confined to a set of Lebesgue measure zero by Lemma~\ref{lem:determining.U.is.generic}, and it is the \emph{reported} RDM rather than the true state that has to satisfy the condition: the statistical noise of the tomography perturbs every reported Pauli expectation value, so $\rhoini$ almost surely lies off that set even when the true RDM lies on it (like $\ket{0}^{\otimes N}$). Note also that block-irreducibility is a distinct condition from non-degeneracy of the spectrum, even though both fail only on measure-zero sets: an RDM diagonal in a product basis with pairwise distinct eigenvalues is non-degenerate but not block-irreducible. In the unlikely event that the reported RDM is not block-irreducible, the tomography can simply be repeated, or another neighbouring qubit chosen to enlarge it.

The lemmas above tell us that $U$ is determined uniquely, up to a global phase, by $\rhoini$ and $\rhotar$. They do not provide an algorithm that reconstructs it. We give such an algorithm in App.~\ref{app:reconstructing.U}. In short, it rests on the observation that Eq.~\eqref{eq:rhotar.main} can be read as a linear condition on the gate, $\rhotar\,(U \otimes \mathbb{1}) = (U \otimes \mathbb{1})\,\rhoini$. Collecting every matrix that obeys the same relation,
\begin{equation} \label{eq:Sspace.main}
 \mathfrak{S} = \left\lbrace S \in M_n(\mathbb{C}) \,\middle\vert\, \rhotar (S \otimes \mathbb{1}) = (S \otimes \mathbb{1}) \rhoini \right\rbrace,
\end{equation}
one obtains a space that contains $U$ and that, for block-irreducible $\rhoini$, is one-dimensional. Any non-zero solution $\hat{U}$ of this linear system is therefore a complex multiple of $U$, and dividing it by its operator norm returns $U$ up to a global phase. Neither RDM is diagonalized at any point.

For gates constructed in this way, emulation of the circuit model by the Motte model becomes straightforward for $k=3$. This is because two-qubit gates ($m=2$) can be exactly reproduced with $k=m+1=3$, allowing the user to simply write in the circuit directly.

\subsubsection{State preparation universality for gates via QITE} \label{sec:stateprep}

Since equivalence with the circuit model has already been demonstrated with gates via density-matrix transformation at $k=3$, for gates via QITE we will instead consider a related problem: efficient state preparation~\cite{ward2009}. Indeed, this is a problem more related to the intended usage of the Motte model in guiding a set of variables through quantum dynamics toward targets.

By construction, QITE with domain $k$ can prepare any state that arises as a ground state of an $\ell$-local Hamiltonian whose imaginary-time trajectory from the initial state stays within correlation length $k$. When $k=2$, this imposes both that the correlation length $C$ and Hamiltonian locality $\ell$ satisfy $C,\ell \lesssim 2$. For states that have larger $C$, inexact QITE is the only option and convergence to the true ground state is not guaranteed. Furthermore, though each round is guaranteed to be polynomial, the number of rounds required to reach convergence may not be efficient. QITE is therefore not a guaranteed and efficient method for arbitrary state preparation as soon as $k< N$ (the standard regime for useful QITE).

However, this statement applies only to QITE as described in Ref.~\cite{motta2020qite}, for which the Hamiltonian is fixed. Instead, the Motte model also allows for the target to change from step to step. For the case of QITE, this means applying the method to a constantly adapting Hamiltonian. We call this adapt-H QITE (unrelated and not to be mistaken with adapt-QITE~\cite{Gomes2021}).

Both in QITE and adapt-H QITE, one infinitesimal step of the process moves $|\psi\rangle$ along the direction $-(h_m - \langle h_m\rangle)|\psi\rangle$: it is the projection of $h_m|\psi\rangle$ onto the tangent space at $|\psi\rangle$, and the direction is the one that minimizes $h_m$ by varying correlations up to a distance of $k$. But even considering only the $k(>1)$-domain around $h_m$, Pauli-string $h_m$ is always multiply degenerate as soon as $h_m$ involves more than one qubit. For standard inexact QITE the set of terms $h_m$ is fixed, and each step performs steepest descent of $\langle h_m\rangle$ on projective Hilbert space (App.~\ref{app:AHQITEU}) using only the descent direction that its domain-$k$ update can realize. Such a confined descent stalls wherever that realizable direction vanishes, and the resulting critical point need not be the ground state: where $h_m$ is degenerate, the directions that would lift the degeneracy and lower the energy further call for correlations beyond the domain $k$, so they fall outside the realizable set and the trajectory halts short of convergence. By varying $h_m$ over all 2-local Hermitian operators on adjacent pairs, the issue does not occur by construction. However, this is at the cost of increased complexity: finding other useful $h_m$ to adjust the adapted imaginary time evolution trajectory and avoid or exit the problematic inexact non-converging trajectories. We consider that the user is providing all useful $h_m$, effectively ignoring the extra complexity cost of finding them in the following. Allowing for all $h_m$, the set of reachable infinitesimal directions is
\begin{equation}
 \text{span}\left\lbrace (P - \langle P\rangle)|\psi\rangle \right\rbrace,
\end{equation}
where $P$ are the 2-local Pauli operators on edges of $G$. For a connected graph $G$ and a generic state $|\psi\rangle$, this span is $\mathcal{O}(N)$-dimensional and does not exhaust the entire tangent space at $|\psi\rangle$ in projective Hilbert space. Instead, universality follows from Lie-bracket closure over composed steps. A more detailed and rigorous demonstration is provided in App.~\ref{app:AHQITEU}.

\subsubsection{Worked example: GHZ state preparation}

As an example of the Motte model emulating a quantum circuit, consider the preparation of a three qubit GHZ state via a Hadamard and two chained CNOT gates.

Specifically, we consider the target state for which
\begin{multline*}
 \langle Z\otimes Z\otimes I \rangle = \langle Z\otimes I\otimes Z \rangle \\ = \langle I\otimes Z\otimes Z \rangle = \langle X\otimes X\otimes X \rangle = 1,
\end{multline*}
the state stabilized by $X\otimes X\otimes X$, $Z\otimes Z\otimes I$ and $I\otimes Z\otimes Z$. The eight elements of the stabilizer group fix the non-zero expectation values: the four above, together with
\begin{equation*}
 \langle X\otimes Y\otimes Y \rangle = \langle Y\otimes X\otimes Y \rangle = \langle Y\otimes Y\otimes X \rangle = -1;
\end{equation*}
all others vanish.

This is done from the initial $\ket{000}$ state, for which
\begin{equation*}
 \langle P_0\otimes P_1 \otimes P_2 \rangle = 1,
\end{equation*}
for any $P_j = I$ or $Z$ and with all other expectation values zero.

We will specifically consider gates via density matrix transformation, and so use a $k=m+1$ qubit tomography to apply an $m$-qubit gate. Since the tomography is exact only up to an accuracy $\delta$, the user will not be given the ideal values as above, but rather approximations that can be assumed to be unique for each Pauli. This matters here: the RDMs of the ideal $\ket{000}$ are diagonal in the computational product basis and so are not block-irreducible, and it is the perturbation of the reported values that puts the reported RDMs in the generic class of Lemma~\ref{lem:determining.U.is.generic} on which the reconstruction relies.

The first gate to recover is the Hadamard on the rightmost qubit. The two-qubit tomography of the resulting state (obtained by analytic computation or performed on a circuit model computer) satisfies:
\begin{equation*}
\forall P \in \{I, X, Y, Z\}, \quad \left\lbrace \begin{array}{l}
 \langle P\otimes Z \rangle \leftrightarrow\langle P\otimes X \rangle, \\
 \langle P\otimes Y \rangle \rightarrow -\langle P\otimes Y \rangle
\end{array}\right.
\end{equation*}
The user therefore manipulates the tomographic expectation values accordingly and uses the result as an input.

Since the tomographic expectation values are all unique, it is evident to the interface computer that the user applied the transformation above, and therefore a Hadamard gate, and adds this gate to the circuit. The swap $\langle X\rangle\leftrightarrow\langle Z\rangle$ with $\langle Y\rangle\to-\langle Y\rangle$ on the rightmost qubit is precisely the conjugation action of the Hadamard ($HXH=Z$, $HYH=-Y$), so the gate can be read off from this signature. Indeed, for the case of Clifford gates, inspection of the expectation values is sufficient to deduce the gate without the need for diagonalization of the RDMs.

The tomography is then rerun, with the results supplied to the user. Again the user receives a set of uniquely valued expectation values, all an approximation up to a $\delta$ error on their true values. They apply the transformations corresponding to a CNOT controlled on the rightmost qubit and targeted on the middle qubit, such that
\begin{equation*}
 \forall P \in \{I, X, Y, Z\}, \quad \left\lbrace \begin{array}{l}
 \langle P\otimes I \otimes X \rangle \leftrightarrow \langle P\otimes X \otimes X \rangle, \\
 \langle P\otimes Z \otimes I \rangle \leftrightarrow \langle P\otimes Z \otimes Z \rangle,
\end{array}\right.
\end{equation*}
with $\langle P\otimes X\otimes I\rangle$ and $\langle P\otimes I\otimes Z\rangle$ left fixed. These are the images of the Pauli generators of the pair, so every remaining value follows as a product.
Again the interface computer uses this to deduce the required transformation and appends the required gate to the circuit. The result is a circuit to construct the Bell pair $(\ket{00}+\ket{11})/\sqrt{2}$ on the middle and rightmost qubit.

After repeating the tomography, the same process is used to apply a CNOT controlled on the middle qubit and targeted on the leftmost. This is done by the user applying transformations such that
\begin{equation*}
 \forall P \in \{I, X, Y, Z\}, \quad \left\lbrace \begin{array}{l}
 \langle I \otimes X \otimes P \rangle \leftrightarrow \langle X \otimes X \otimes P \rangle, \\
 \langle Z \otimes I \otimes P \rangle \leftrightarrow \langle Z \otimes Z \otimes P \rangle,
\end{array}\right.
\end{equation*}
with $\langle X\otimes I\otimes P\rangle$ and $\langle I\otimes Z\otimes P\rangle$ left fixed,
and the interface computer again deducing and appending the required transformation.

The same process would be applied for the case of a four-qubit GHZ state and higher. The fact that the number of qubits would exceed $k=3$ would not be an impediment since $k=m+1$ is sufficient to uniquely determine an $m$-qubit gate.

\subsection{Polynomial overhead}
\label{sec:overhead}
The cost of emulating a circuit through the Motte model is an overhead relative to direct execution, arising from the tomography that precedes each layer of gates. We assume throughout this section that the problem graph coincides with, or is contained in, the chip coupling graph $G$, so that every gate is native to $G$ and no routing is required. If not, the routing overhead must be added to the bound provided. $D$ is then the depth of the circuit being emulated.

The circuit is built one layer at a time. Before each of the $\mathcal{O}(D)$ depth increments, a system-wide $k$-body tomography is performed, specified by $S(N,k)$ measurement settings. Estimating every $\le k$-body Pauli expectation value to additive accuracy $\delta$ requires $\mathcal{O}(\delta^{-2}\,k\log N)$ shots per setting, by Hoeffding's inequality~\cite{hoeffding1963} together with a union bound over the at most $\mathcal{O}(N^k)$ estimated values. The cost of the emulation relative to a single execution of the circuit is therefore the product of the $\mathcal{O}(D)$ rounds, the $S(N,k)$ settings, and the shots per setting,
\begin{equation}\label{eq:overhead}
 \mathrm{Overhead} = \mathcal{O}\!\big(D\,S(N,k)\,\delta^{-2}\,k\log N\big).
\end{equation}

The number of settings depends on the connectivity of $G$, and is far smaller than the number of observables because Paulis on disjoint regions commute and share a setting. Reconstructing all Paulis on a single $k$-qubit region takes $3^k$ settings, one per assignment of a measurement basis to its $k$ qubits. On a graph of bounded degree $\Delta$ there are $\mathcal{O}(N)$ connected $k$-qubit regions, but each overlaps only $\mathcal{O}(k\Delta^{k-1})$ others, so they partition into $\mathcal{O}(k\Delta^{k-1})$ groups of mutually disjoint regions that are measured together. The number of settings is therefore
\begin{equation}\label{eq:settings}
 S(N,k)=3^k\,\mathcal{O}(k\Delta^{k-1}).
\end{equation}
When accounting for qubit-wise Pauli-grouping~\cite{verteletskyi2020clique}, its extension to fully commuting groups via unitary basis changes~\cite{yen2020compatible}, or more advanced techniques like classical shadows~\cite{huang2020shadows}, this bound can drop even further, at a classical computing cost and extra imprecisions. We do not need these refinements to show universality, so we ignore them here. On a fully connected graph (i.e.\ $\Delta=N$) the regions no longer have bounded overlap: for $k=2$, $\mathcal{O}(\log N)$ settings are required~\cite{garcia2020pairwise}, and more generally $\mathcal{O}(f(k)\log N)$ with $f(k)=e^{\mathcal{O}(k)}$~\cite{cotler2020overlapping}. Such connectivity can be relevant for some problems and is possible in some cold atoms- or ions-based hardware (for which they therefore likely display better performances than their competitors). For simplicity's sake and because it does not change the qualitative conclusions, we do not consider fully connected or even dense graphs in the rest of this paper, we use only Eq.~\eqref{eq:settings} below; present solid-based QPU coupling maps (linear, heavy-hex, grid) have bounded degree $\Delta=\mathcal{O}(1)$, and we take $k$ fixed.

For gates via density-matrix transformation, taken here at $k\geq 3$, the gate applied to the device is fixed exactly by the tomographic report, independently of $\delta$ (this gate-correctness property is established in Appendix~\ref{app:gatecorrect}), so a fixed accuracy $\delta=\mathcal{O}(1)$ suffices and the statistical noise of the tomography does not propagate into the emulated state. A single execution of the depth-$D$ circuit costs $\mathcal{O}(ND)$ elementary operations; the work complexity of the emulation is this multiplied by the overhead~\eqref{eq:overhead}. Substituting the bounded-degree setting count~\eqref{eq:settings} with $\delta=\mathcal{O}(1)$, we find the work-complexity bound:
\begin{equation}\label{eq:cost}
\begin{aligned}
 \mathrm{Work} &= \underbrace{\mathcal{O}(ND)}_{\text{one execution}}\times
 \underbrace{\mathcal{O}\!\big(D\,S(N,k)\,\delta^{-2}\,k\log N\big)}_{\text{overhead}} \\
 &= \mathcal{O}\!\big(3^k\,k^2\,\Delta^{k-1}\;D^2\,N\log N\big).
\end{aligned}
\end{equation}
The two factors of $D$ are the depth of each tomographed circuit and the number of layers; $\Delta$ is the degree of the coupling graph; the leading $N$ is the register width; and the $\log N$ is the per-setting shot overhead derived above. For fixed $k$ and bounded $\Delta$ the prefactor $3^k k^2\Delta^{k-1}$ is a constant independent of $N$, so the work complexity is
\begin{equation}\label{eq:cost-collapsed}
 \mathcal{O}\!\big(D^2\,N\log N\big).
\end{equation}
This is polynomial, so BQP membership is preserved.

This overhead is the price of emulating the circuit model through the Pauli interface; it is not a performance penalty on algorithms designed natively in the model. When compiling a circuit into the extended $k=2$ gate set, the entangling gate is instead obtained by eigenstate transformation, whose accuracy is limited by the spectral gap of the reduced density matrix; the resulting error accumulates over the circuit, and reaching a fixed target fidelity $\epsilon$ costs $\mathcal{O}\!\big(D^4 N^3\log N/(\epsilon^2\gamma_{\min}^2)\big)$ operations, with $\gamma_{\min}$ the smallest gap encountered (Appendix~\ref{app:k2overhead}). The exact $k=3$ construction avoids this.

A final component of the cost is classical, borne by the interface computer, and we account for it here only to confirm that it does not dominate. The interface manipulates exclusively $2^k$-dimensional objects, and what it does with them depends on the route: for gates via eigenstate transformation it reconstructs the local reduced density matrices from the reported expectation values, diagonalizes them and pairs their eigenvectors; for gates via density-matrix transformation it solves instead the linear system~\eqref{eq:Sspace.main} for the gate itself, with neither RDM diagonalized; and for gates via QITE it solves a linear system over the Paulis of the $k$-qubit domain. It never instantiates the $2^N$-dimensional statevector, since the tomography of the circuit so far is performed by the device rather than simulated classically. Each such operation costs $e^{\mathcal{O}(k)}$ for the $\mathcal{O}(N)$ gates of a layer, plus $\mathcal{O}(N\log N)$ to aggregate the shots, so the classical cost over all $D$ layers is $\mathcal{O}\!\big(D\,N\,(e^{\mathcal{O}(k)}+\log N)\big)$ in time and $\mathcal{O}\!\big(N\,e^{\mathcal{O}(k)} + ND\big)$ in space. For fixed $k$ this is polynomial in $N$ and $D$ and is dominated by the quantum overhead of Eq.~\eqref{eq:cost-collapsed}, which is why we disregard it in the remainder of the paper. It would turn exponential only if $k$ were allowed to grow with $N$, as in exact QITE at a diverging correlation length, a regime that inflates the quantum cost in the same step.

In summary, the cost of emulating a circuit through the Motte model for qubits in a lattice with fixed accuracy $\delta=\mathcal{O}(1)$ and $3 \leq k = \mathcal{O}(1)$ is $\mathcal{O}\!\big(D^2\,N\log N\big)$. For $k=2$ with extension, the cost of compiling a circuit is $\mathcal{O}\!\big(D^4 N^3\log N/(\epsilon^2\gamma_{\min}^2)\big)$, with $\epsilon$ the fixed target fidelity and $\gamma_{\min}$ the smallest gap encountered.

\subsection{Effects of noise}

In considering noise we must distinguish between \emph{tomographic noise} and \emph{state-evolution noise}. Tomographic noise has two components. \emph{Shot noise} is the statistical sampling error $\delta$: expectation values are inexact for any finite number of samples, but it averages down as the number of shots grows and, as we noted in Sec.~\ref{subsec:dens}, can even be beneficial, by lifting degeneracies for the eigenstate-transformation route and by moving the reported RDM into the block-irreducible class for the density-matrix route. \emph{Readout noise} is a systematic measurement error that causes 0 to be misreported as 1 and vice-versa; unlike shot noise it does not average away (but it can be mitigated~\cite{bravyi2021mitigating}). Both components preserve \emph{gate correctness} exactly; their effect falls on \emph{intention correctness}, two concepts we now define. State-evolution noise covers effects such as gate imperfections and decoherence.

Gate correctness concerns whether the gate applied in the circuit is the one which the user requested, given the tomography results they were given and the targets they supplied. Intention correctness concerns whether the gate applied does what the user wants.

An important property of the Motte model is that tomographic noise affects only intention correctness, and not gate correctness: even if the tomography supplied to the user is misleading due to noise, if the corresponding targets differ by a unitary $U$, then for gates via density-matrix transformation that $U$ is recovered exactly by the interface computer and applied to the QPU, regardless of the tomographic accuracy (Appendix~\ref{app:gatecorrect}). For gates via eigenstate transformation the recovery is exact only up to the phase ambiguity of Sec.~\ref{sec:emulation_eigenstates}.

Both forms of noise do have a deleterious effect on intention correctness. An extreme example would be that the state is already that desired by the user, but the tomography incorrectly reports a false initial state. The user would then supply the target expectation values and gates would be applied to the QPU accordingly. These then have the unintended effect of moving the state away from the target.

Though this work focuses on matters of complexity and universality, and so neglects state-evolution noise for the most part, it is worth noting that the Motte model can be expected to have a degree of robustness against it. Since the intended use-case is to continually update the user's representation of the state and drive toward a target, it allows the computation to be continually corrected even when noise steers it away. The idealized model, with its exact gate correctness, presumes a noiseless device and so is realized in full only on fault-tolerant hardware. On NISQ hardware gate correctness still holds against the tomographic stage (sampling and readout noise alike), and the re-aiming loop confers partial robustness to state-evolution noise, so use on NISQ hardware remains appropriate.

\subsection{Summary}

To summarize this section, let us consider the standard Motte model to be that for which all three methods of gate addition are available to the user. Gates via density-matrix transformation then render the model universal for BQP at $k=3$: every two-qubit gate is reproduced exactly, and, by the gate-correctness property, the emulation is independent of the statistical accuracy $\delta$ of the intermediate tomography. When the problem graph is native to a coupling graph of bounded degree, each layer's tomography requires a number of measurement settings that is constant in $N$ for fixed $k$, and the total cost is $\mathcal{O}(D^2 N\log N)$ elementary operations.

Universality is also reached at $k=2$ by combining gates via density-matrix transformation, which implement the single-qubit gates exactly, with gates via eigenstate transformation, which supply the entangling gate up to a known unitary at the cost of extending the Motte model to allow faithful report of this unitary on demand (App.~\ref{app:combined}). The entangling gate then carries a statistical error of order $\delta$, so reaching a fixed target fidelity requires $\delta$ to decrease polynomially with $N$ and $D$. For the variants in which only gates via eigenstate transformation, or only gates via (non-adapt-H) QITE, are allowed, the value of $k$ at which full circuit-model equivalence is achieved is unknown. Equivalence and universality are also unknown in the combined $k=2$ case without extension, and so is the equivalence for the combined $k=2$ case with extension.

\section{Applications of the Motte Model}
\label{sec:applications}

Now that we have defined the model and demonstrated its computational power, we will consider some examples of applications.

\subsection{Quantum imaginary time evolution}

The QITE algorithm of Ref.~\cite{motta2020qite} can now be considered an application of the Motte model. QITE finds
ground and thermal states of a Hamiltonian $H=\sum_m h_m$
by simulating imaginary-time evolution, requiring exponentially less
space and time per iteration than classical ITE and avoiding
the deep ancilla-heavy circuits needed by phase estimation. Each step
drives the state toward the ground state via a small correction unitary
derived from current Pauli expectation values. It is used mainly in chemistry and nuclear physics~\cite{Yeter-Aydeniz2019, Gomes2020, Nishi2021}.

\subsection{Quantum-native procedural generation}

\texttt{QuantumGraph}~\cite{wootton2020map} demonstrated the same
idiom in a creative context. To procedurally generate geopolitical
maps in the style of \emph{Civilization}, one qubit is assigned per
nation. The expectation values $\langle X\rangle$, $\langle
Y\rangle$, $\langle Z\rangle$ encode policy dimensions (aggression,
expansion, defense); the Bloch-sphere constraint $\sum_i\langle
\sigma_i\rangle^2\leq 1$ automatically enforces a policy budget
without explicit classical rules. Pairwise correlators encode
inter-nation relationships, and the monogamy of entanglement ensures
that nations in conflict concentrate decisions on each other. Between
turns, quantum gates update the Pauli values; the final map is sampled
from the projective measurement distribution. This is precisely the
Motte model interface, predating its formal statement. The same
approach has since been applied to level generation in the game
\emph{Quantum Backrooms}~\cite{fromholz2026backrooms}.

Both applications have been demonstrated on real quantum hardware at scales exceeding 50 qubits.
Examples include the map-generation algorithm of \texttt{QuantumGraph} that was executed on the 53-qubit IBM Rochester device in 2020~\cite{wootton2020map}, and the implementation of \emph{Quantum Backrooms} on the \texttt{ibm\_miami} device in 2026~\cite{fromholz2026backrooms}.
These demonstrations establish that the Motte model idiom is not merely a theoretical construct. It is also a software paradigm that has been exercised on state-of-the-art hardware from the early NISQ era to the present day.

\subsection{A shared structure: control by fractional steps}

Beyond sharing a vocabulary, QITE and \texttt{QuantumGraph} share a concrete operational structure the model makes precise. Both perform a $k$-body tomography at each step and derive their next gates from its result, and both advance the state by a \emph{fractional} amount: the principal power $U^f$, $f\in[0,1]$, of Sec.~\ref{sec:def}. In QITE this fraction is the inverse step count, $f=1/n$ (equivalently $\Delta\tau=\beta f$); in \texttt{QuantumGraph} it is the partial gate applied between turns. The fraction $f$ is thus the common control of both algorithms, smaller $f$ buying a more faithful trajectory at the cost of more rounds.

Because each gate is derived from a preceding tomography, both inherit the model's \emph{gate correctness} (App.~\ref{app:gatecorrect}): the gate applied is the one dictated by the tomographic picture and the supplied target, not a noisier approximation of it, so the construction compounds no error of its own and the measure-decide-apply loop re-aims at every step. This faithfulness is exact and $f$-independent when gates are built by the density-matrix route ($k\ge3$); for QITE and \texttt{QuantumGraph}, small $f$ refines the trajectory but not the gate, whose accuracy is set by the relevant spectral gap rather than by $f$. The resulting robustness, an account of why QITE tolerates the noise in its tomography and why \texttt{QuantumGraph}'s tomography-driven updates remain stable, is a property the model isolates and proves, not one either algorithm was originally formulated to possess.

\section{Conclusion}
\label{sec:conclusion}

The Motte model is a formally universal, BQP-equivalent model of
quantum computation whose defining feature is that target expectation values are specified
and state information is handled entirely in terms of Pauli
expectation values, quantities that are both physically natural and
hardware-accessible. Though dependent on tomography to add gates to the circuit, a key property which we call \emph{gate correctness} is that the unitary applied by the device is determined by the tomographic report in a well-defined way: exactly, and independently of the tomographic accuracy, via the density-matrix route at $k=3$; at $k=2$ only up to a residual statistical perturbation of the eigenbasis, which accumulates over depth. At $k=3$ the iteration is therefore insensitive to tomographic-stage noise, statistical sampling plus moderate readout error, and the re-aiming loop confers partial robustness to state-evolution noise in the underlying gates and dynamics, whose characterization lies beyond the present scope. The price is a polynomial overhead relative to bare circuit execution, a consequence of mandatory per-depth tomography.

The model formalizes an idiom that has already produced working
algorithms, namely QITE for quantum chemistry and \texttt{QuantumGraph} for
procedural generation. Together they provide the conceptual framework needed to
extend this idiom as hardware scales toward the fault-tolerant era.

Beyond the two motivating applications described here, concrete use-cases
for the Motte model remain sparse. This is less a reflection of the model's
limitations than of the hardware era in which its idiom was developed:
\texttt{QuantumGraph} was demonstrated on a 53-qubit device at a time when
such devices were exceptional, but interesting applications of the
interface would require a large range of Pauli expectation values to work from. That situation has now changed. QPUs of 50 or even 100 qubits and beyond
are routinely available to all via cloud services, and devices of this scale are
being operated as a matter of course. The Motte model provides a principled interface
for exploring what can be done with this hardware: one whose defining properties,
hardware-native Pauli constraints, gate correctness under tomographic noise, and a
vocabulary of physically observable quantities, are well matched to the devices now available.
The search for applications will now begin in earnest.

\begin{acknowledgments}
The authors would like to thank Spencer Topel and João Ferreira for discussions and for exploring use-cases of \texttt{QuantumGraph}. JRW would like to thank Junye Huang for discussions. The name of the model honors Motta \emph{et al.}~\cite{motta2020qite}, whose QITE algorithm is one of the primary motivations for the model, as well as the name of the company where this line of work continues: \emph{Moth}.

\end{acknowledgments}

\section*{Data Availability}
No datasets were generated or analyzed in this work. The quantum imaginary time evolution algorithm that motivates the model is described in Ref.~\cite{motta2020qite}, and the \texttt{QuantumGraph} framework is publicly available at Ref.~\cite{quantumgraph}.

\section*{Author Contributions}
The authors' contributions are described following the CRediT taxonomy. \textbf{Conceptualization:} J.R.W. \textbf{Methodology:} J.R.W., M.I.-M., P.F. \textbf{Formal analysis:} M.I.-M., P.F. \textbf{Software:} J.R.W., P.F. \textbf{Investigation:} all authors. \textbf{Validation:} M.I.-M., P.F. \textbf{Visualization:} P.F. \textbf{Writing - original draft:} all authors. \textbf{Writing - review \& editing:} all authors. \textbf{Supervision:} J.R.W., P.F. \textbf{Project administration:} J.R.W., P.F. \textbf{Funding acquisition:} J.R.W. \textbf{AI/LLM use:} the authors used Anthropic's Claude (Claude Code, Claude Opus~4 and~5 families) to assist with \LaTeX{} editing and prose revision; all output was reviewed and verified by the authors, who take full responsibility for the contents of this manuscript.

\section*{Conflict of Interest}
All authors are full-time employees of Moth, which develops its own versions of \texttt{QuantumGraph} and the QITE algorithm discussed in this work. Some of the authors additionally hold a minor equity interest in the company. The authors declare no other conflicts of interest.

\bibliography{refs}

\appendix

\section{Overhead of the \texorpdfstring{$k=2$}{k=2} entangling gate}\label{app:k2overhead}

At $k=2$ (extended scheme of App.~\ref{app:combined}) the entangling gate is obtained by eigenstate transformation: the interface pairs the eigenvectors of the tomography-inferred initial and target reduced density matrices by eigenvalue. Writing the reported matrix as $\rhoini=\rho+E$, with $\rho$ the exact RDM and $\lVert E\rVert=\mathcal{O}(\delta)$, the Davis-Kahan $\sin\Theta$ theorem~\cite{daviskahan1970} bounds the rotation of an eigenvector whose eigenvalue is separated from the rest of the spectrum by a gap $\gamma$: its deviation is $\mathcal{O}(\lVert E\rVert/\gamma)$. With $\gamma_{\min}$ the smallest such gap encountered along the construction, the applied entangling gate therefore differs from the intended one by $\mathcal{O}(\delta/\gamma_{\min})$ in operator norm.

Unitary errors compose sub-additively, so over the $M=\mathcal{O}(ND)$ entangling gates of the compiled circuit the total error is $\mathcal{O}(ND\,\delta/\gamma_{\min})$. Imposing a target fidelity $\epsilon$ fixes $\delta=\mathcal{O}(\epsilon\,\gamma_{\min}/ND)$, hence $\mathcal{O}\!\big((ND/\epsilon\gamma_{\min})^2\,k\log N\big)$ shots per setting. The work complexity is again $\mathcal{O}(ND)$ times the overhead,
\begin{equation}\label{eq:k2cost}
\begin{aligned}
 \mathrm{Work}_{k=2} &= \underbrace{\mathcal{O}(ND)}_{\text{one execution}} \\
 &\quad\times\underbrace{\mathcal{O}\!\big(D\,S(N,2)\,\delta^{-2}\,k\log N\big)}_{\text{overhead}} \\
 &= \mathcal{O}\!\left(\frac{D^4 N^3\log N}{\epsilon^2\,\gamma_{\min}^2}\right),
\end{aligned}
\end{equation}
up to the $k$-dependent prefactor. This bound holds in the perturbative regime $\delta\ll\gamma_{\min}$; the smallest gap $\gamma_{\min}$ is a property of the emulated trajectory and may in principle be small, in which case the eigenvector pairing is no longer well conditioned. The exact $k=3$ construction of App.~\ref{app:reconstructing.U}, whose correctness is established in App.~\ref{app:gatecorrect}, is free of this dependence: it never pairs eigenvectors, and so never sees $\gamma_{\min}$.

\section{Reconstruction of \texorpdfstring{$U$}{U} from the RDMs} \label{app:reconstructing.U}

Suppose we have a unitary $U$ supported on $m$ qubits and the RDM $\rhoini$ of these qubits together with one extra qubit, taken to be a neighbor on the coupling graph $G$ so that the $(m{+}1)$-body tomography is available. Let
\begin{equation} \label{eq:rhotar.def}
 \rhotar := (U \otimes \mathbb{1})\, \rhoini \, (U \otimes \mathbb{1})^{\dagger}
\end{equation}
be the RDM after applying $U \otimes \mathbb{1}$ to the $m+1$ qubits, matching Eq.~\eqref{eq:rhotar.main} of the main text. In this section, we show that, if $\rhoini$ is block-irreducible, then we can reconstruct $U$, up to a global phase, from $\rhoini$ and $\rhotar$. The proof consists of Lemma~\ref{lem:uniquely.determining.U} and~\ref{lem:determining.U.is.generic}, which were stated in Sec.~\ref{subsec:dens}. We repeat the statements here for convenience. We recall that $n = 2^m$, so that $\rhoini$ and $\rhotar$ are $2n \times 2n$ matrices.

\begin{lemma*}[Restatement of Lemma~\ref{lem:uniquely.determining.U}]
 If there exists no proper subspace $\{0\} \neq E \subsetneq \mathbb{C}^{n}$ such that $\rhoini$ sends $E \otimes \mathbb{C}^2$ and $E^{\perp} \otimes \mathbb{C}^2$ to themselves, that is, if $\rhoini$ is block-irreducible, then $\rhoini$ and $\rhotar$ determine $U$ up to a global phase.
\end{lemma*}

\begin{lemma*}[Restatement of Lemma~\ref{lem:determining.U.is.generic}]
 The set of hermitian matrices that fail to be block-irreducible, that is, that preserve $E \otimes \mathbb{C}^2$ and $E^{\perp} \otimes \mathbb{C}^2$ for some proper subspace $\{0\} \neq E \subsetneq \mathbb{C}^n$, has Lebesgue measure zero. Block-irreducibility, and hence the hypothesis of Lemma~\ref{lem:uniquely.determining.U}, therefore holds generically.
\end{lemma*}

\begin{proof}[Proof of Lemma~\ref{lem:uniquely.determining.U}]
 \smallskip
 \noindent\textbf{Step 1: Setting up notation.}

 To exploit the tensor product with $\mathbb{C}^2$, we set up notation to discuss blocks in the matrices $\rhoini$ and $U$. For $i,j \in \{1,2\}$, let $E_{i,j}$ be the $2 \times 2$-matrix that has entries $0$ except at row $i$ and column $j$, where it has entry $1$. Since $\rhoini$ is a $2n \times 2n$-matrix, we can split it into four blocks, each of which is an $n \times n$-matrix. More formally, we can write
 \[ \rhoini = \sum_{i,j=1}^2 \rho_{i,j} \otimes E_{i,j}, \]
 which defines the blocks $\rho_{i,j}$ referred to below.

 \smallskip
 \noindent\textbf{Step 2: Two distinct unitaries $U_1$ and $U_2$ yield the same $\rhotar$ if and only if $U_2^{\dagger} U_1$ commutes with $\rho_{i,j}$ for all $i,j \in \{ 1, 2 \}$.}
 
 Let $U_1$ and $U_2$ be two unitaries such that
 \[ (U_1 \otimes \mathbb{1}) \rhoini (U_1 \otimes \mathbb{1})^{\dagger} = (U_2 \otimes \mathbb{1}) \rhoini (U_2 \otimes \mathbb{1})^{\dagger}. \]
 Multiplying on the left by $(U_2 \otimes \mathbb{1})^{\dagger}$ and on the right by $(U_2 \otimes \mathbb{1})$, and setting $V = U_2^{\dagger} U_1$, this is equivalent to
 \[ (V \otimes \mathbb{1}) \rhoini (V \otimes \mathbb{1})^{\dagger} = \rhoini, \]
 that is, $U_1$ and $U_2$ yield the same $\rhotar$ if and only if
 \[ [V \otimes \mathbb{1}, \rhoini] = 0. \]
 Writing this in block matrices, we get for every $i,j \in \{ 1, 2 \}$
 \[ [V, \rho_{i,j}] = 0. \]
 This finishes Step 2.

 \smallskip
 \noindent\textbf{Step 3: There exists no $\{0\} \neq E \subsetneq \mathbb{C}^n$ as in the statement of the Lemma, if and only if the only unitaries $V$ that commute with each $\rho_{i,j}$ are of the form $V = e^{i\theta} \cdot \mathbb{1}$ for some $\theta \in \mathbb{R}$.}
 
 We show the equivalence of the negatives instead. Suppose there is $\{0\} \neq E \subsetneq \mathbb{C}^n$ such that $\rhoini$ sends $E \otimes \mathbb{C}^2$ to itself and $E^{\perp} \otimes \mathbb{C}^2$ to itself. Let $P_E$ and $P_{E^{\perp}}$ denote the orthogonal projections in $\mathbb{C}^n$ to $E$ and $E^{\perp}$ respectively. Then $\rhoini$ commutes with $P_E \otimes \mathbb{1}$ and with $P_{E^{\perp}} \otimes \mathbb{1}$. Writing this in block-matrices, we get that $\rho_{i,j}$ commutes with $P_E$ and $P_{E^{\perp}}$ for all $i,j$. In particular, the matrix $V := P_E - P_{E^{\perp}}$ is a unitary matrix and commutes with each $\rho_{i,j}$. Since $E$ is a proper subspace, $V$ has eigenvalues both $1$ and $-1$. Thus, it cannot be written as $e^{i \theta} \cdot \mathbb{1}$. However, $(V \otimes \mathbb{1}) \rhoini (V \otimes \mathbb{1})^{\dagger} = \rhoini$. In other words, we cannot distinguish the identity on $\mathbb{C}^n$ from $V$ based on $\rhoini$ and $\rhotar$ alone.

 Suppose now there exists a unitary $V$ such that $V$ commutes with every $\rho_{i,j}$ and $V$ is not equal to $e^{i\theta} \cdot \mathbb{1}$ for some $\theta \in \mathbb{R}$. Then $V$ has at least two distinct eigenvalues with eigenspaces $E_1$, $E_2$ that are orthogonal to each other. Since $V$ commutes with $\rho_{i,j}$, $\rho_{i,j}$ has to preserve the eigenspaces of $V$. In particular, $\rho_{i,j}$ has to preserve the orthogonal complement of $E_1$, since this orthogonal complement is the sum of all the other eigenspaces. We conclude that $E_1$ is a proper subspace of $\mathbb{C}^n$ such that $\rhoini$ preserves $E_1 \otimes \mathbb{C}^{2}$ and $E_1^{\perp} \otimes \mathbb{C}^2$. This proves the equivalence and the Lemma.
\end{proof}

\begin{proof}[Proof of Lemma~\ref{lem:determining.U.is.generic}]
 Let $e_1, \dots, e_n$ be the standard basis of $\mathbb{C}^n$. Fix $1 \leq r < n$ and let $E_r$ be the subspace spanned by $e_1, \dots, e_r$. A hermitian $\rho$ preserves $E_r \otimes \mathbb{C}^2$ and $E_r^{\perp} \otimes \mathbb{C}^2$ precisely when it is block-diagonal with respect to their direct sum,
 \begin{equation*}
 \rho = \rho_{E_r} \oplus \rho_{E_r^{\perp}},
 \end{equation*}
 with $\rho_{E_r}$ hermitian on the $2r$-dimensional space $E_r \otimes \mathbb{C}^2$ and $\rho_{E_r^{\perp}}$ hermitian on the $2(n-r)$-dimensional space $E_r^{\perp} \otimes \mathbb{C}^2$. As a real vector space, the space of hermitian $d \times d$-matrices has dimension $d^2$ (inside the $2d^2$-dimensional space of complex $d \times d$-matrices). Thus, the set of $\rho$ of the form above has real dimension $4r^2 + 4(n-r)^2$.

 We now want to look at all $r$-dimensional subspaces $E$. These can be parametrized by the complex Grassmannian $G(r,n)$, whose real dimension is $2r(n-r)$. We obtain that the real dimension of a space containing all $\rho$ that violate Lemma~\ref{lem:uniquely.determining.U} through an $r$-dimensional $E$ is at most
 \[ 4r^2 + 4(n-r)^2 + 2r(n-r) = 4n^2 - 6r(n-r). \]
 In contrast, the space of all $2n \times 2n$ hermitian matrices has real dimension $(2n)^2 = 4n^2$. For $1 \leq r \leq n-1$ we have $6r(n-r) \geq 6(n-1) > 0$, since $n = 2^m \geq 2$, so the set of all $\rho$ that preserve an $r$-dimensional subspace and its orthogonal complement is the image of a smooth map from a manifold of strictly smaller dimension. The differential of such a map is nowhere surjective, so every point of its domain is critical, and Sard's theorem gives that the image has Lebesgue measure zero. Going over all $1 \leq r < n$, we have a finite union of measure zero sets. This has measure zero as well, which finishes the proof.
\end{proof}

We end this section with a procedure that explicitly reconstructs $U$ from a block-irreducible $\rhoini$ and the corresponding $\rhotar$. Rearranging Eq.~\eqref{eq:rhotar.def} shows that the intended gate obeys the linear relation $\rhotar\,(U \otimes \mathbb{1}) = (U \otimes \mathbb{1})\,\rhoini$. We therefore consider the space of all matrices obeying that same relation,
\begin{equation} \label{eq:Sspace}
 \mathfrak{S} := \left\lbrace S \in M_n(\mathbb{C}) \,\middle\vert\, \rhotar (S \otimes \mathbb{1}) = (S \otimes \mathbb{1}) \rhoini \right\rbrace,
\end{equation}
which contains $U$ and so is not trivial. Let $\hat{U} \in \mathfrak{S} \setminus \{ 0 \}$. Substituting Eq.~\eqref{eq:rhotar.def} into Eq.~\eqref{eq:Sspace} and multiplying on the left by $(U \otimes \mathbb{1})^{\dagger}$, we get
\[ \rhoini \left( (U^{\dagger}\hat{U}) \otimes \mathbb{1} \right) = \left( (U^{\dagger}\hat{U}) \otimes \mathbb{1} \right) \rhoini. \]
In other words, $U^{\dagger} \hat{U}$ is in the commutant
\begin{equation} \label{eq:Kspace}
 \mathfrak{K} := \left\lbrace S \in M_n(\mathbb{C}) \,\middle\vert\, (S \otimes \mathbb{1}) \rhoini = \rhoini (S \otimes \mathbb{1}) \right\rbrace.
\end{equation}

As we will show (Lemma~\ref{lem:Ker.is.1.dimensional}), $\mathfrak{K} = \mathbb{C} \cdot \mathbb{1}$ whenever $\rhoini$ is block-irreducible. Since $\mathfrak{S} = U \cdot \mathfrak{K}$, the space $\mathfrak{S}$ is then one-dimensional as well, so the linear system defining it has a solution space spanned by $U$ alone. In consequence, we get the following algorithm to obtain $U$.

\begin{enumerate}
 \item Find a non-zero element of $\mathfrak{S}$ by solving the system of linear equations~\eqref{eq:Sspace} that defines it. This yields a matrix $\hat{U}$, unique up to a complex scale.
 \item By Lemma~\ref{lem:Ker.is.1.dimensional}, $\hat{U} = \lambda U$ for some $\lambda \in \mathbb{C}^{*}$. Since $U$ is unitary its operator norm is $1$, so $\lVert \hat{U} \rVert = \lvert \lambda \rvert$ and
 \[ V := \frac{\hat{U}}{\lVert \hat{U} \rVert} = \frac{\lambda}{\lvert \lambda \rvert}\, U = e^{i\theta} U \]
 for some $\theta \in [0;2\pi[$: $V$ is $U$ up to a global phase. Dividing by $\lvert \det \hat{U} \rvert^{1/n}$ serves equally well.
\end{enumerate}

Lemma~\ref{lem:Ker.is.1.dimensional} makes $\mathfrak{S}$ one-dimensional for every block-irreducible $\rhoini$, but as $\rhoini$ approaches the block-reducible set the system~\eqref{eq:Sspace} becomes increasingly ill-conditioned, a practical obstacle to its numerical implementation. Refinements of the algorithm addressing this would be beneficial, but lie outside the scope of the present work.

We now show Lemma~\ref{lem:Ker.is.1.dimensional}.
\begin{lemma} \label{lem:Ker.is.1.dimensional}
 Let $\rhoini$ be block-irreducible. If $S \in \mathfrak{K}$ and $S \neq 0$, then $S \in \mathbb{C} \cdot \mathbb{1}$.
\end{lemma}

\begin{proof}
 Let $S \in \mathfrak{K}$ such that $S \neq 0$.

 \smallskip
 \noindent\textbf{Step 1: Show that $S$ is invertible.}
 By definition of $\mathfrak{K}$, $[S \otimes \mathbb{1}, \rhoini] = 0$. Writing $\rhoini$ in block matrices as in the proof of Lemma~\ref{lem:uniquely.determining.U}, we get
 \[ [S, \rho_{i,j}] = 0 \]
 for all $i, j \in \{ 1, 2 \}$. This implies that, for all $i, j$, $\rho_{i,j}$ preserves $\ker(S)$, since for every $v \in \ker(S)$:
 \[ 0 = S\rho_{i,j}(v) - \rho_{i,j}S(v) = S\rho_{i,j}(v).\]
 Therefore, $\rhoini$ preserves $\ker(S) \otimes \mathbb{C}^2$ and, since $\rhoini$ is hermitian, it also preserves $\ker(S)^{\perp} \otimes \mathbb{C}^2$. Since $\rhoini$ is block-irreducible, this implies that either $S = 0$, which we exclude, or $\ker(S) = 0$, in which case $S$ is invertible.

 \smallskip
 \noindent\textbf{Step 2: Show that $S \in \mathbb{C} \cdot \mathbb{1}$.}
 Since $S$ is invertible, $S \otimes \mathbb{1}$ is invertible. Let $\lambda$ be an eigenvalue of $S$ and thus non-zero. Since multiples of $\mathbb{1}$ commute with any matrix and $S \in \mathfrak{K}$, we get that
 \[ [S - \lambda \mathbb{1}, \rho_{i,j}] = 0 \]
 for all $i,j$. Using the same argument as in Step 1, this implies that $\ker(S - \lambda \mathbb{1})$ has to be either $0$ or $\mathbb{C}^n$. Since $\lambda$ is an eigenvalue, the kernel cannot be zero and we conclude that $S - \lambda \mathbb{1} = 0$. Therefore, $S = \lambda \mathbb{1} \in \mathbb{C} \cdot \mathbb{1}$. This finishes the proof of the Lemma.
\end{proof}

\section{Combined gates at \texorpdfstring{$k=2$}{k=2}} \label{app:combined}

Universality can be achieved via arbitrary single-qubit gates and a single entangling two-qubit gate~\cite{brylinski2002universal} at the cost of extending the Motte model. This can be achieved via the Motte model for $k=2$ if gates via eigenstate transformation and gates via density matrix transformation are combined and if we allow the reporting of a single 2-qubit entangling gate deduced once by the Motte model, but reused identically irrespectively of the subsequent reported initial states (which is not possible in the above defined Motte model). Gates via density matrix transformation allow perfect implementations of any single-qubit gates by enlarging the tomography to a neighboring qubit on the coupling graph (see Sec.~\ref{subsec:dens}), providing the required single-qubit gates. That neighbor is passive: it is not prepared in any particular state and is used only to enlarge the tomographic RDM, and the constructed two-qubit unitary acts as $U\otimes\mathbb{1}$ across the partition, so no uncomputation is required. Two-qubit gates then use the eigenstate transformation method, which fixes the entangler only up to the phase ambiguity of Sec.~\ref{sec:emulation_eigenstates}, in an eigenbasis that is itself perturbed by statistical noise. If we allow the extension of the Motte model device to report to the user the unitaries it applies, the phase ambiguity is resolved: the reported unitary is a fixed, known, generically entangling two-qubit gate, which supplies the required entangling resource. The report does not remove the statistical perturbation of the eigenbasis, however: the reported gate still deviates from the intended one by an amount set by the tomographic accuracy $\delta$ and the spectral gap of the reduced density matrix, an error that accumulates over depth and is quantified in Sec.~\ref{sec:overhead} and App.~\ref{app:k2overhead}. Some relative phases can be optionally corrected using the single-qubit gate perfect implementation. In other words, the user deduces once a 2-qubit entangling gate for which they cannot fully choose (using gates via eigenstate transformation), but they then can reuse this exact same gate (the extension), along with any 1-qubit gates that they can fully choose (via density matrix transformation) to build any circuit they want, achieving universality.

This becomes more precise if we allow for transformations to be set as a template and reused. For example, if the initial state is $\langle Z\otimes I\rangle=\langle I\otimes X\rangle(=\langle Z\otimes X\rangle)=1$ and the target is $\langle Z\otimes Z\rangle=\langle X\otimes X\rangle=1$, the transformation will have the canonical effect of the CNOT with variations due to tomography and phase. Allowing the user to repeat this whenever a similar effect is desired allows for a fixed and known entangled gate (up to the residual statistical jitter), for which constructive techniques for universality are known~\cite{bremner2002}. Like the eigenstate-transformation hack of Sec.~\ref{sec:emulation_eigenstates}, this remains a proof-by-construction: the user is hijacking the Motte interface to obtain a known unitary rather than using the model as intended.

Without the extension, the 2-qubit gates obtained using gates via eigenstate transformation change uncontrollably for each new reported RDM, making the entangling resource adaptative. We do not know of, nor propose a proof of universality in this case. We also do not know of, nor propose a proof of the equivalence between this combined gates at $k=2$ model, with or without the extension, and the circuit model. We have performed tests on simulator with a small number of qubits that seem to suggest both the universality and the equivalences, however.

\section{Gate correctness of the density-matrix construction}\label{app:gatecorrect}

For gates added via density-matrix transformation, the unitary applied to the device is exactly the one intended by the user, independently of the statistical accuracy of the tomography. The algebra needed for this is already in place: the statement below is Lemma~\ref{lem:uniquely.determining.U} read in the operational setting of Sec.~\ref{subsec:dens}, and what this appendix adds to it is the observation that its conclusion never relies on the reported RDM being accurate.

\begin{proposition}\label{prop:gatecorrect}
Let the user intend an $m$-qubit unitary $U$, applied as $U\otimes\mathbb{1}$ on the $m$ target qubits together with one passive neighboring qubit included only to enlarge the tomography. Let $\rhoini$ be the reported $(m{+}1)$-qubit reduced density matrix and $\rhotar=(U\otimes\mathbb{1})\,\rhoini\,(U\otimes\mathbb{1})^\dagger$ the target the user computes classically from $\rhoini$. If $\rhoini$ is block-irreducible in the sense of Lemma~\ref{lem:uniquely.determining.U}, then any gate $U_c\otimes\mathbb{1}$ the interface reconstructs from $(\rhoini,\rhotar)$ equals $U$ up to a global phase, independently of the tomographic accuracy $\delta$.
\end{proposition}

\begin{proof}
The reconstructed gate reproduces the target, so the state it actually produces from $\rhoini$ is the target itself,
\begin{equation*}
 \rhofin:=(U_c\otimes\mathbb{1})\,\rhoini\,(U_c\otimes\mathbb{1})^\dagger=\rhotar .
\end{equation*}
Both $U$ and $U_c$ therefore carry $\rhoini$ to the same matrix, so by Step~2 of the proof of Lemma~\ref{lem:uniquely.determining.U} the unitary $U_c^\dagger U$ commutes with every block of $\rhoini$, and by Step~3 block-irreducibility forces it to be a scalar. Hence $U_c=U$ up to a global phase. By Lemma~\ref{lem:determining.U.is.generic} the hypothesis fails only on a set of Lebesgue measure zero, which statistical tomographic noise almost surely lifts (and which, if ever encountered, is cleared by repeating the tomography or choosing another neighbor).

What is specific to gate correctness is that this argument invokes only the relation $\rhotar=(U\otimes\mathbb{1})\,\rhoini\,(U\otimes\mathbb{1})^\dagger$, never the value of $\rhoini$: a noisy report $\rhoini=\rho+E$, with $\rho$ the exact RDM and $\lVert E\rVert=\mathcal{O}(\delta)$, is matched by a target the user builds on that same $\rhoini$, so the perturbation is common to both descriptions. The reconstruction is therefore exact regardless of the statistical accuracy $\delta$.
\end{proof}

Tomographic noise thus affects only which state the user believes the device to hold (intention correctness), not the unitary that is applied (gate correctness). The distinction between $\rhotar$ and $\rhofin$ above is exactly this one: gate correctness is the statement that $\rhofin=\rhotar$, and it holds however far the reported $\rhoini$ has drifted from the true state of the device.

The eigenstate-transformation route carries different caveats, and they are of a different kind. There the interface does diagonalize both RDMs and pair their eigenvectors, so it does require non-degenerate spectra: generic, and in the presence of statistical noise almost surely satisfied, a residual degeneracy being lifted by the noise itself or removed by repeating the tomography, and it does require the user and the interface to diagonalize the same reported RDM under the same ordering convention, which holds by construction, the interface having supplied that RDM to the user. Its accuracy is then limited by the spectral gap rather than being exact, whence the $\gamma_{\min}$ of App.~\ref{app:k2overhead}.

\section{Adapt-H QITE's universality} \label{app:AHQITEU}
We give a rigorous argument for the universality of adapt-H QITE, introduced in Sec.~\ref{sec:stateprep}: with two-local domain ($k=2$) on a connected coupling graph it can prepare any pure state. The mechanism is the one anticipated in the main text (Lie-bracket closure of low-dimensional generators) but realized at the level of the flows the algorithm actually generates.

\subsection{Definitions and notation}

We encourage the reader to skip entirely these standard fundamental definitions, rephrased in the notations used in the text, and to return to them only if they are needed in the rest of this section. They are included here for convenience and to fully understand the orbit theorem, a quantum control result crucial for the demonstration.

\paragraph{Hilbert space and its K\"ahler structure.}
Let $\HH=\mathbb{C}^{\,n+1}$ be the system Hilbert space, of complex dimension $n+1$, equipped with the Hermitian inner product $\braket{\cdot}{\cdot}$ (linear in the ket, antilinear in the bra). Regarded as a \emph{real} vector space, $\HH$ has dimension $2(n+1)$ and carries three mutually compatible structures, obtained by splitting the inner product into its real and imaginary parts: writing $g(u,v)=\mathrm{Re}\braket{u}{v}$ and $\omega(u,v)=\mathrm{Im}\braket{u}{v}$,
\begin{equation*}
 \braket{u}{v}=g(u,v)+i\,\omega(u,v).
\end{equation*}
Here $g$ is a real, symmetric, positive-definite bilinear form (the \emph{Riemannian metric}) and $\omega$ is a real antisymmetric form (the \emph{symplectic form}). The third structure is the \emph{complex structure} $J$, the real-linear operator given by multiplication by $i$, $J u = i\,u$, $J^2=-\mathbb{1}$. It ties the other two together through $\omega(u,v)=g(Ju,v)$ and $g(Ju,Jv)=g(u,v)$; any two of $(g,\omega,J)$ determine the third, and together they constitute the \emph{K\"ahler structure} of $\HH$.

\paragraph{Projective Hilbert space.}
A \emph{pure state} is a ray: an equivalence class of nonzero vectors under multiplication by a nonzero complex scalar, $v\sim\lambda v$ ($\lambda\in\mathbb{C}^\ast$). The set of rays is the \emph{projective Hilbert space} $\mathbb{P}(\HH)$; once the canonical basis is fixed it is the complex projective space $\mathbb{CP}^{\,n}$, of complex dimension $n$ (two real dimensions are removed from $\HH$, one by normalization and one by the global phase). It is compact, connected, and simply connected. The K\"ahler structure of $\HH$ descends to $\mathbb{P}(\HH)$, which is therefore a K\"ahler manifold; in particular it carries the \emph{Fubini-Study metric} $g$. For $N$ qubits, $\HH=(\mathbb{C}^2)^{\otimes N}$, so $n+1=2^N$ and $\mathbb{P}(\HH)=\mathbb{CP}^{\,2^N-1}$.

\paragraph{Tangent and horizontal space.}
The \emph{tangent space} $T_pM$ at a point $p$ of a manifold $M$ is the vector space of velocities of curves through $p$. For $M=\mathbb{P}(\HH)$, at a unit vector $\ket\psi$ we use the \emph{horizontal} representative
\begin{equation*}
 T_\psi=\{\ket{u}\in\HH:\braket{\psi}{u}=0\},
\end{equation*}
the complex-orthogonal complement of $\ket\psi$, with orthogonal projector $P_\psi^\perp=\mathbb{1}-\ket\psi\bra\psi$ onto it. It is $J$-invariant (a complex subspace of complex dimension $n$, real dimension $2n$); the removed complex line $\mathbb{C}\ket\psi=\mathrm{span}_{\mathbb R}\{\ket\psi,J\ket\psi\}$ is exactly the normalization and phase (gauge) directions.

\paragraph{Vector fields and flows.}
A \emph{vector field} $X$ on $M$ assigns to each $p$ a tangent vector $X(p)\in T_pM$. Its \emph{integral curves} solve $\dot p(t)=X(p(t))$, and the \emph{time-$t$ flow} $\Phi^t_X$ sends $p$ to the point reached after following $X$ for time $t$ ($\Phi^0_X=\mathrm{id}$, $\Phi^{s+t}_X=\Phi^s_X\circ\Phi^t_X$). The \emph{forward} flow uses $t>0$, the \emph{backward} flow $t<0$, with $\Phi^{-t}_X=(\Phi^t_X)^{-1}=\Phi^t_{-X}$. A field is \emph{complete} if its flow exists for all $t$, automatic on a compact manifold such as $\mathbb{P}(\HH)$. A \emph{family of vector fields} $\mathcal D$ is a set of vector fields on $M$.

\paragraph{Derivative and gradient.}
For a map $F$ defined near $p$, its \emph{directional derivative} at $p$ along $u\in T_pM$ is $DF_p[u]=\tfrac{d}{dt}\big|_{0}F(\gamma(t))$, on any curve $\gamma$ with $\gamma(0)=p$, $\dot\gamma(0)=u$; it is linear in $u$. For a real-valued $E$ this is the \emph{differential} $dE_p(u)\equiv DE_p[u]$, a real-linear functional on $T_pM$; for a vector field $Y$, $DY_p[u]$ is the derivative of $p\mapsto Y(p)$, the object entering the bracket $[X,Y]=DY[X]-DX[Y]$. The \emph{gradient} $\nabla E$, with respect to the Riemannian metric $g$, is the unique tangent vector with
\begin{equation*}
 g(\nabla E,u)=dE_p(u)\qquad\text{for all }u\in T_pM,
\end{equation*}
$-\nabla E$ is the steepest-descent direction.

\paragraph{Lie bracket and Lie algebra.}
A \emph{Lie algebra} is a vector space with a bilinear, antisymmetric bracket $[\cdot,\cdot]$ obeying the Jacobi identity $[X,[Y,Z]]+[Y,[Z,X]]+[Z,[X,Y]]=0$. Two instances occur below: for operators, $[A,B]=AB-BA$; for vector fields, $[X,Y]=DY[X]-DX[Y]$. The \emph{Lie algebra generated by} $\mathcal D$, written $\mathrm{Lie}(\mathcal D)$, is the smallest space of vector fields containing $\mathcal D$ and closed under linear combinations and the bracket. Its elements are vector \emph{fields}; the \emph{evaluated Lie algebra} at $p$ is obtained by plugging in that point,
\begin{equation*}
 \mathrm{Lie}_p(\mathcal D)=\{X(p):X\in\mathrm{Lie}(\mathcal D)\}\subseteq T_pM ,
\end{equation*}
a subspace of the finite-dimensional tangent space $T_pM$.

\paragraph{Control system and orbit.}
A \emph{driftless control-affine system} is $\dot p=\sum_a c_a(t)\,f_a(p)$, where the $f_a$ are fixed \emph{control vector fields} and the $c_a(t)\in\mathbb{R}$ are chosen freely (``driftless'': no term independent of the controls). It is \emph{symmetric} if its set of control fields is closed under sign change ($-f_a$ is also a control field), so any motion can be reversed. The \emph{orbit} of $p$ is the set reachable by composing the forward and backward flows of the control fields,
\begin{multline}
 \mathcal \mathcal{O}(p)=\\
 \bigl\{(\Phi^{\,t_k}_{X_k}\!\circ\cdots\circ\Phi^{\,t_1}_{X_1})(p):
 k\in\mathbb{N},\ X_i\in\mathcal D,\ t_i\in\mathbb{R}\bigr\};
\end{multline}
for a symmetric system the reachable set coincides with the orbit. The orbits partition $M$; an orbit is \emph{open} (resp.\ \emph{closed}) when it is an open (resp.\ closed) subset of $M$. If every orbit is open then each is also closed (its complement is a union of the other open orbits), so on a \emph{connected} $M$ there is a single orbit, equal to all of $M$.

\paragraph{Distributions, rank, and analyticity.}
A \emph{distribution} assigns to each $p$ a subspace of $T_pM$ (here $p\mapsto\mathrm{Lie}_p(\mathcal D)$); its \emph{rank} at $p$ is that subspace's dimension. The distribution is \emph{full} at $p$ if $\mathrm{Lie}_p(\mathcal D)=T_pM$ (rank $\dim M$), and \emph{full} (equivalently \emph{bracket-generating}) if this holds everywhere. The \emph{constant-rank} hypothesis demands a $p$-independent rank which is a regularity assumption many smooth-category theorems require, but \emph{not} needed in the real-analytic category. A function is \emph{real-analytic} ($C^\omega$) if near each point it equals its convergent Taylor series (stronger than smooth, $C^\infty$). An \emph{immersion} is a smooth map $\iota:N\to M$ with injective differential everywhere (not necessarily globally injective); an \emph{immersed submanifold} is the image of an injective immersion which is intrinsically a manifold, but need not carry the subspace topology of $M$ (it may wind around densely), unlike an \emph{embedded} submanifold.
\subsection{A QITE step is a gradient flow}
One infinitesimal QITE step on a Hermitian term $h$ moves $\ket\psi$ along
\begin{align}\label{eq:Vh}
 V_h(\psi)\;&\equiv\;-\bigl(h-\langle h\rangle\bigr)\ket\psi
 \;=\;-\,P_\psi^\perp h\ket\psi, \\
 \langle h\rangle &=\bra\psi h\ket\psi . \nonumber
\end{align}
This is (minus) the Fubini-Study gradient of the energy $E(\psi)=\langle h\rangle$ on $\mathbb{P}(\HH)$, i.e.\ its steepest-descent direction~\cite{brockett1991,stokes2020natgrad}. Equivalently, $V_h$ is the velocity of the normalized imaginary-time flow $e^{-\tau h}\ket\psi/\lVert e^{-\tau h}\ket\psi\rVert$, so one QITE step advances the state by an infinitesimal imaginary-time step (up to a constant rescaling of the flow rate).

To verify both claims, recall the Fubini-Study metric, the real inner product induced from $\HH$ on the horizontal tangent space, $g_\psi(u,v)=\mathrm{Re}\braket{u}{v}$ for $u,v\in T_\psi$. The energy $E(\psi)=\langle h\rangle$ is phase-invariant, hence descends to $\mathbb{P}(\HH)$; its differential along a horizontal $\ket u$ is
\begin{align*}
 dE_\psi(u)
 &=\braket{u}{h\psi}+\braket{h\psi}{u}
 =2\,\mathrm{Re}\braket{h\psi}{u}\\
 &=2\,\mathrm{Re}\bigl\langle (h-\langle h\rangle)\psi\,\big|\,u\bigr\rangle,
\end{align*}
the last step using $\mathrm{Re}\braket{\langle h\rangle\psi}{u}=\langle h\rangle\,\mathrm{Re}\braket{\psi}{u}=0$ for horizontal $\ket u$. Comparing with $g_\psi(\nabla E,u)=dE_\psi(u)$ for all $\ket u\in T_\psi$ identifies the gradient, and hence $V_h$:
\begin{equation*}
 \nabla E(\psi)=2\,(h-\langle h\rangle)\ket\psi,
 \qquad
 V_h(\psi)=-\tfrac12\,\nabla E(\psi),
\end{equation*}
so $V_h$ points along $-\nabla E$ (the factor $\tfrac12$ rescales the rate, not the direction). Independently, differentiating the normalized flow $\ket{\psi(\tau)}=e^{-\tau h}\ket{\psi_0}/\lVert e^{-\tau h}\ket{\psi_0}\rVert$ at $\tau=0$, using $\tfrac{d}{d\tau}e^{-\tau h}\ket{\psi_0}\big|_0=-h\ket{\psi_0}$ and $\tfrac{d}{d\tau}\lVert e^{-\tau h}\ket{\psi_0}\rVert\big|_0=-\langle h\rangle$, gives
\begin{equation*}
 \tfrac{d}{d\tau}\ket{\psi(\tau)}\big|_0
 =-(h-\langle h\rangle)\ket{\psi_0}
 =V_h(\psi_0).
\end{equation*}
The two flows thus coincide up to a constant rescaling of time, both monotonically decreasing $\langle h\rangle$.

The QITE construction returns a Hermitian $A(h)$ fixed by its first-order action on the current state, $-iA(h)\ket\psi=-(h-\langle h\rangle)\ket\psi=V_h(\psi)$ (the defining property of the imaginary-time step; see the QITE gates of Sec.~\ref{sec:model}). The unitary $e^{-i\Delta\tau A(h)}$ therefore reproduces this direction,
\begin{equation*}
 e^{-i\Delta\tau A(h)}\ket\psi=\ket\psi+\Delta\tau\,V_h(\psi)+\mathcal \mathcal{O}(\Delta\tau^2),
\end{equation*}
so ``apply a QITE step with term $h$'' is the same as ``flow for time $\Delta\tau$ along $V_h$''. The admissible terms are $h\in\mathcal L_2$, the real vector space of one- and two-local Hermitian operators on the edges of $G$. Because $\mathcal L_2$ is closed under negation, $-h\in\mathcal L_2$ and $V_{-h}=-V_h$: \emph{both} signs of every control direction are available.

\subsection{A single step is not enough}
As $h$ ranges over $\mathcal L_2$, the directions $V_h(\psi)$ fill
\begin{align*}
 \mathcal R_\psi=\{V_h(\psi):h\in\mathcal L_2\}, \\
 \dim_{\mathbb R}\mathcal R_\psi\le\dim\mathcal L_2=\mathcal \mathcal{O}(N),
\end{align*}
a vanishing fraction of $\dim T_\psi=2(2^N-1)$. In particular there is \emph{no} choice of $h$ that realises a generic target direction in one step; universality must come from \emph{composing} non-commuting steps.

\subsection{Universality via geometric control}
Fix a basis $h_1,\dots,h_m$ of $\mathcal L_2$, with $m=\dim\mathcal L_2=\mathcal \mathcal{O}(N)$. Since $V_h$ is linear in $h$ (from \eqref{eq:Vh}), $V_{\sum_a c_a h_a}=\sum_a c_a V_{h_a}$, such that adapt-H QITE is the driftless control-affine system on $\mathbb{P}(\HH)$
\begin{equation}\label{eq:control}
 \dot\psi=\sum_{a=1}^{m}c_a(t)\,V_{h_a}(\psi),\qquad c_a(t)\in\mathbb R,
\end{equation}
with control fields $\{V_{h_a}\}$. It is \emph{symmetric}, since $\mathcal L_2$ is closed under negation and $V_{-h}=-V_h$, so each $-V_{h_a}$ is again a control field; hence the set reachable from $\ket{\psi_0}$ equals its orbit.

The orbit is characterized by the orbit theorem, in the notation above ($\mathcal D$ a family of vector fields, $\mathrm{Lie}(\mathcal D)$ its generated Lie algebra, $\mathrm{Lie}_p(\mathcal D)\subseteq T_pM$ its evaluation, $\mathcal \mathcal{O}(p)$ the orbit).

\medskip
\noindent\textbf{Orbit theorem (Sussmann; Hermann-Nagano)~\cite{jurdjevic1972,dalessandro2021,schirmer2001}.}
\emph{For any family $\mathcal D$ of complete vector fields, each orbit $\mathcal \mathcal{O}(p)$ is an immersed connected submanifold of $M$ whose tangent space at every $q\in\mathcal \mathcal{O}(p)$ contains the evaluated Lie algebra, $T_q\mathcal \mathcal{O}(p)\supseteq\mathrm{Lie}_q(\mathcal D)$. If $M$ and the fields of $\mathcal D$ are real-analytic, this is an equality, $T_q\mathcal \mathcal{O}(p)=\mathrm{Lie}_q(\mathcal D)$.}

\medskip
\noindent We apply this with $M=\mathbb{P}(\HH)$ and $\mathcal D=\{V_{h_a}\}$; the generated Lie algebra is $\mathfrak F$, with evaluated Lie algebra $\mathfrak F(\psi)\subseteq T_\psi$. The fields are complete ($\mathbb{P}(\HH)$ is compact) and polynomial in the components of $\ket\psi$, hence real-analytic, so the analytic case applies, $T_\psi(\mathrm{orbit})=\mathfrak F(\psi)$, with no separate constant-rank hypothesis. Consequently, if the distribution is full, $\mathfrak F(\psi)=T_\psi$ for every $\ket\psi$, then by the open/closed-orbit argument above ($\mathbb{P}(\HH)$ being connected) there is a single orbit, equal to all of $\mathbb{P}(\HH)$: every pure state is reachable from any $\ket{\psi_0}$. It thus suffices to establish $\mathfrak F(\psi)=T_\psi$ at every $\ket\psi$.

Applying the complex structure (multiplication by $i$) to a gradient field $V_A(\psi)$ (following the definition Eq.~\eqref{eq:Vh} for $A$ Hermitian) gives the corresponding Hamiltonian (Schr\"odinger) field $iV_A(\psi)$.

\begin{lemma}[Bracket structure]\label{lem:brackets}
For Hermitian $A,B$, writing $C=-i[A,B]$ (Hermitian), the bracket of the gradient fields, projected onto $T_\psi$, is
\begin{equation*}
 P_\psi^\perp [V_A,V_B]=i\,V_C .
\end{equation*}
\end{lemma}

\begin{proof}
Write $V_A(\psi)=-(A-\langle A\rangle)\ket\psi$ and likewise for $V_B$. The directional derivative of $V_B$ along $\ket u$ is
\begin{multline*}
 DV_B(\psi)[u]= \\ -(B-\langle B\rangle)\ket u
 +\bigl(\braket{u}{B\psi}+\braket{\psi}{Bu}\bigr)\ket\psi .
\end{multline*}
Forming $[V_A,V_B]=DV_B[V_A]-DV_A[V_B]$, every term proportional to $\langle A\rangle$, $\langle B\rangle$, $\langle A\rangle\langle B\rangle$, and to the anticommutator $\langle\{A,B\}\rangle$ appears identically in both pieces and cancels, leaving
\begin{equation*}
 [V_A,V_B](\psi)=-iC\ket\psi .
\end{equation*}
Projecting onto $T_\psi$ gives $-i(C-\langle C\rangle)\ket\psi=iV_C(\psi)$.
\end{proof}

Lemma~\ref{lem:brackets} shows that the assignment $A\mapsto V_A$ intertwines the operator commutator with the (projected) vector-field bracket. Because $T_\psi$ is $J$-invariant, $i$ commutes with $P_\psi^\perp$, and for these fields the bracket's departure from complex linearity lies along $\ket\psi$, hence drops out under $P_\psi^\perp$; the $i$ therefore factors through, giving also $P_\psi^\perp[iV_A,V_B]=-V_C$ and $P_\psi^\perp[iV_A,iV_B]=-iV_C$. Let $\mathfrak{g}$ (later shown to be $\mathfrak{su}(2^N)$) be the Lie algebra generated by $\mathcal L_2$ under commutation; then:

\begin{corollary}\label{cor:closure}
$\mathfrak F$ contains $V_X$ for every $X\in\mathfrak{g}$.
\end{corollary}

\begin{proof}
By the relations above, bracketing fields of type $V$ or $iV$ (on $\mathbb{P}(\HH)$) returns a field of the same two types, whose operator label is the corresponding commutator-generator $-i[\cdot,\cdot]$. We use this together with the fact that every generator is itself a commutator of generators: $-i[X_j,Y_j]\propto Z_j$, and on an edge $(j,k)$, $-i[Z_j,Y_jX_k]\propto X_jX_k$ (and similarly for every $1$- and $2$-local Pauli), so $\mathcal L_2\subseteq\{-i[h,h']:h,h'\in\mathcal L_2\}$.

\emph{Step 1 ($iV_X\in\mathfrak F$ for all $X\in\mathfrak g$).} By $[V_h,V_{h'}]=i\,V_{-i[h,h']}$, $\mathfrak F$ contains $iV_{-i[h,h']}$ for all $h,h'\in\mathcal L_2$; by the fact above these labels already include every generator, so $iV_h\in\mathfrak F$ for all $h\in\mathcal L_2$. By $[iV_A,iV_B]=-\,i\,V_{-i[A,B]}$, the set $\{X:iV_X\in\mathfrak F\}$ is closed under commutation; being a subalgebra that contains $\mathcal L_2$, it contains all of $\mathfrak g$.

\emph{Step 2 ($V_X\in\mathfrak F$ for all $X\in\mathfrak g$).} With every $iV_X$ now available, the identity $[iV_X,V_Y]=-V_{-i[X,Y]}$ shows that $\{X:V_X\in\mathfrak F\}$ is also closed under commutation: if $V_X,V_Y\in\mathfrak F$ then (as $iV_X\in\mathfrak F$) so is $V_{-i[X,Y]}$. This subalgebra contains the generators $\mathcal L_2$, hence equals $\mathfrak g$.
\end{proof}

For two-local Paulis on a connected graph $G$, $\mathfrak{g}=\mathfrak{su}(2^N)$: this is the Lie-algebraic content of the universality of one- and two-qubit gates~\cite{lloyd1996universal}. Concretely, brackets such as $\bigl[[X_jY_k,Z_kZ_l],Y_jX_k\bigr]\propto Z_jZ_l$ reach non-adjacent pairs and, on iteration, all Pauli strings.

\begin{lemma}[Spanning]\label{lem:span}
$\{V_X(\psi):X\in\mathfrak{su}(2^N)\}=T_\psi$ for every $\ket\psi$.
\end{lemma}

\begin{proof}
Fix $\ket v\in T_\psi$, so $\braket{\psi}{v}=0$ (note, $\ket v$ is not necessarily normalized). Set $X=\ket v\bra\psi+\ket\psi\bra v$. Then $X$ is Hermitian and traceless, hence $X\in\mathfrak{su}(2^N)$; moreover $X\ket\psi=\ket v$ and $\langle X\rangle=0$, so $V_X(\psi)=-\ket v$. Every tangent vector is thus realized.
\end{proof}

Combining Corollary~\ref{cor:closure} and Lemma~\ref{lem:span}, $\mathfrak F$ contains $\{V_X:X\in\mathfrak{su}(2^N)\}$, whose evaluations span $T_\psi$ at every $\ket\psi$. The distribution is therefore full, and the orbit through any $\ket{\psi_0}$ is all of $\mathbb{P}(\HH)$. We conclude:

\medskip
\noindent\textbf{Proposition (universality).}
\emph{Adapt-H QITE with two-local domain $(k=2)$ on a connected coupling graph $G$ is universal for state preparation: for any pure states $\ket{\psi_0},\ket{\psi_\star}$ there is a finite sequence of adaptively chosen terms $h_1,\dots,h_n\in\mathcal L_2$ whose QITE steps carry $\ket{\psi_0}$ to $\ket{\psi_\star}$, exactly for the idealized continuous flow~\eqref{eq:control}, and to arbitrary accuracy for its finite-$\Delta\tau$ implementation.}

\medskip
\noindent Both kinds of finiteness are immediate: the orbit theorem reaches $\ket{\psi_\star}$ by a \emph{finite} concatenation of flows (exactly), and at any fixed accuracy each segment is realized by finitely many $\Delta\tau$-steps, the first-order integrator error scaling as $\mathcal{O}(\Delta\tau)$; only their \emph{number} is at issue, and is bounded now.

\medskip
\noindent The measure-decide-apply loop is, in control-theoretic language, the same reachability problem solved by engineered-dissipation and measurement-driven state preparation~\cite{verstraete2009,kraus2008,ticozzi2008}, in which a designed, here, adaptively chosen, generator drives the state toward a target; the universality statement above is its closed-system, gradient-flow instance.

\subsection{Cost: universality is not efficiency}

The quantum work complexity of adapt-H QITE follows closely the formula derived in Sec.~\ref{sec:overhead} and App.~\ref{app:k2overhead}. It must, however, account for the effect of Trotterization and the imprecisions it introduces, which worsen the bound. We now show where.

Emulating a circuit of $M=\mathcal{O}(ND)$ native gates in depth $D$, the work complexity is
\begin{equation}\label{eq:qite-master}
\begin{split}
 \mathrm{Work} = &\underbrace{M}_{\text{gates}}\cdot
 \underbrace{\mu}_{\text{steps/gate}}\cdot
 \underbrace{\mathcal{D}}_{\text{depth}} \\
 &\cdot\underbrace{S(N,k)}_{\text{settings}}\cdot
 \underbrace{R}_{\text{shots/setting}},
\end{split}
\end{equation}
where $M\mu$ is the volume of the emulated circuit (the cost of one execution), $\mathcal{D}=\mathcal{O}(D\mu)$ is the number of non-parallelizable micro-steps (each closed by one system-wide tomography), $S(N,k)$ is the setting count of Eq.~\eqref{eq:settings}, and $R$ the shots per setting. For $\mu=1$ this is exactly the eigenstate-transformation overhead of App.~\ref{app:k2overhead}.

Two factors are specific to QITE. First, we take a single two-qubit gate to be reproducible to synthesis error $\xi$ in $\mathcal{O}(1/\xi)$ micro-steps; sharing an end-to-end error $\xi$ among the $M$ gates assigns each the budget $\xi/M$, so
\begin{equation}\label{eq:qite-mu}
 \mu = \mathcal{O}(M/\xi) = \mathcal{O}(ND/\xi).
\end{equation}

Second, every micro-step solves the QITE linear system~\cite{motta2020qite} from a $k$-body tomography of statistical accuracy $\delta$. As in App.~\ref{app:k2overhead}, the finite-$\delta$ solve perturbs the realized generator by $\mathcal{O}(\delta/\sigma_{\min})$, where $\sigma_{\min}$ is the smallest singular value of that system, the QITE counterpart of the spectral gap $\gamma_{\min}$ there. Since a micro-step applies $e^{-i\Delta\tau A}$ with $\Delta\tau=\mathcal{O}(\xi/M)$, the per-gate error is independent of $\Delta\tau$ at $\mathcal{O}(\delta/\sigma_{\min})$, and accumulating over the $M$ gates a target accuracy fixes $\delta=\mathcal{O}(\xi\sigma_{\min}/ND)$, the relation of App.~\ref{app:k2overhead} under $\gamma_{\min}\to\sigma_{\min}$. The shot count of Sec.~\ref{sec:overhead} is then $R=\mathcal{O}(\delta^{-2}k\log N)=\mathcal{O}(N^2D^2 k\log N/\xi^2\sigma_{\min}^2)$.

Substituting into Eq.~\eqref{eq:qite-master},
\begin{equation}\label{eq:qite-cost}
\begin{split}
 \mathrm{Work}
 &= \mathcal{O}\!\left(\frac{3^k k^2\Delta^{k-1}\,N^5 D^6\log N}{\xi^4\,\sigma_{\min}^2}\right) \\
 &= \mathcal{O}\!\left(\frac{N^5 D^6\log N}{\xi^4\,\sigma_{\min}^2}\right)
\end{split}
\end{equation}
for fixed $k$ and bounded degree $\Delta$: polynomial in $N$ and $D$ whenever $\sigma_{\min}$ is inverse-polynomial. The extra factor $M^2=N^2D^2$ and the two extra powers of accuracy relative to the $\epsilon^{-2}$ of App.~\ref{app:k2overhead} are the price of Trotterization: each gate expands into $\mu=\mathcal{O}(M/\xi)$ micro-steps, and the algorithm depth $\mathcal{D}=\mathcal{O}(D\mu)$ carries this factor a second time. With $\sigma_{\min}$ playing the role of App.~\ref{app:k2overhead}'s $\gamma_{\min}$, the two bounds differ only by these Trotter factors. It can be considered to be the cost of using a circuit-based model in the back of a Motte model in the context of QITE. It follows that:
\begin{itemize}
 \item States preparable by polynomial-depth circuits, along trajectories of bounded correlation length and with $\sigma_{\min}$ inverse-polynomial, admit polynomial-cost adapt-H QITE protocols.
 \item Haar-typical states require sequences of exponential length, or drive $\sigma_{\min}$ exponentially small, and cannot be prepared efficiently.
\end{itemize}
The bounded-$k$ Motte parameter and QITE's domain $k$ together set the per-step cost (polynomial in $N$, exponential in $k$); the total cost is set by the trajectory.

\end{document}

%% file: fig-motte-schematic.tex
\begin{tikzpicture}[
    >={Stealth[length=5pt]},
    font=\footnotesize,
    box/.style={
        rounded corners=2pt,
        draw,
        thick,
        minimum width=3.4cm,
        minimum height=8.5mm,
        align=center,
        inner sep=2pt
    },
    userbox/.style={box, fill=blue!8},
    ifacebox/.style={box, fill=orange!10, minimum width=5.0cm},
    qpubox/.style={box, fill=green!10},
    flowarr/.style={->, line width=0.9pt, shorten >=1pt, shorten <=1pt},
    arclbl/.style={
        font=\scriptsize\itshape,
        align=center,
        fill=white,
        inner sep=2pt,
        text width=1.7cm
    },
    extlbl/.style={
        font=\scriptsize\bfseries\itshape,
        align=center,
        text width=1.5cm,
        inner sep=2pt
    }
]

  \node[userbox]  (user)  at (0,0)    {\textbf{User}\\\scriptsize{(classical)}};
  \node[ifacebox] (iface) at (0,-2.8) {\textbf{Motte interface}\\\scriptsize{(classical)}};
  \node[qpubox]   (qpu)   at (0,-5.6) {\textbf{QPU}\\\scriptsize{gate model on $G$}};

  \node[extlbl] (goal)  at (-3.1, 0)    {end\\target};
  \node[extlbl] (init)  at (-3.1,-5.6)  {initial\\state};
  \node[extlbl] (final) at ( 3.1,-5.6)  {final\\sampling};

  \draw[flowarr] (goal.east)  -- (user.west);
  \draw[flowarr] (init.east)  -- (qpu.west);
  \draw[flowarr] (qpu.east)   -- (final.west);

  \draw[flowarr] (user.south east) to[bend left=65]
        node[arclbl] {target $k$-body\\Paulis\\\mbox{($f$ optional)}} (iface.north east);
  \draw[flowarr] (iface.south east) to[bend left=65]
        node[arclbl] {transpiled\\gates on\\edges of $G$} (qpu.north east);
  \draw[flowarr] (qpu.north west)  to[bend left=65]
        node[arclbl] {shot\\statistics} (iface.south west);
  \draw[flowarr] (iface.north west) to[bend left=65]
        node[arclbl] {$\le k$-body\\$\langle$Pauli$\rangle$\\\mbox{($\pm\delta$)}} (user.south west);

\end{tikzpicture}